\documentclass[11pt]{article}
\usepackage{graphicx}
\usepackage{color}
\usepackage{amsmath, amssymb, amsthm, ascmac, booktabs, caption, comment, enumerate, bm, lineno}
\usepackage{latexsym, mathabx, mathrsfs, mathtools, psfrag, setspace, here, subcaption}
\usepackage{natbib}
\usepackage{diagbox, tikz}
\mathtoolsset{showonlyrefs=true}

\definecolor{mycolor}{cmyk}{0.82, 0.23, 0.26, 0}
\usepackage[colorlinks=true, linkcolor=mycolor, filecolor=mycolor, urlcolor=mycolor, citecolor=mycolor, backref=page]{hyperref}
\usepackage[margin = 2.2cm]{geometry}

\usepackage{titlesec}
\titleformat*{\section}{\Large\bfseries\sffamily}
\titleformat*{\subsection}{\large\bfseries\sffamily}
\titleformat*{\subsubsection}{\large\bfseries\sffamily}

\renewcommand{\hat}{\widehat}
\renewcommand{\tilde}{\widetilde}
\renewcommand{\check}{\widecheck}
\renewcommand{\bar}{\overline}
\newcommand{\bbE}{\mathbb{E}}
\newcommand{\bbP}{\mathbb{P}}

\numberwithin{equation}{section}
 
\allowdisplaybreaks[2]

\theoremstyle{definition}
\newtheorem{theorem}{Theorem}
\newtheorem{assumption}{Assumption}
\newtheorem{definition}{Definition}

\newtheorem{example}{Example}
\newtheorem{lemma}{Lemma}[section]
\newtheorem{proposition}{Proposition}
\newtheorem{remark}{Remark}

\numberwithin{figure}{section}
\numberwithin{table}{section}

\title{Estimating Dyadic Treatment Effects with Unknown Confounders}

\author{
    Tadao Hoshino\thanks{School of Political Science and Economics, Waseda University, 1-6-1 Nishi-waseda, Shinjuku-ku, Tokyo 169-8050, Japan. Email: \href{mailto:thoshino@waseda.jp}{thoshino@waseda.jp}.} \ 
    and Takahide Yanagi\thanks{Graduate School of Economics, Kyoto University, Yoshida Honmachi, Sakyo, Kyoto, 606-8501, Japan. Email: \href{mailto:yanagi@econ.kyoto-u.ac.jp}{yanagi@econ.kyoto-u.ac.jp}.}
}

\begin{document}
	
\maketitle
	
\begin{abstract}
	This paper proposes estimation and inference methods for assessing treatment effects with dyadic data.
	Under the assumption that the treatments follow an exchangeable distribution, our approach allows for the presence of any unobserved confounding factors that potentially cause endogeneity of treatment choice without requiring additional information other than the treatments and outcomes.
	Building on the literature of graphon estimation in network data analysis, we propose a neighbourhood kernel smoothing method for estimating dyadic average treatment effects, and derive the rate of convergence of the proposed estimator under certain regularity conditions.
	We also develop conformal inference methods for predicting outcomes conditional on treatment status.
	We apply our methods to international trade data to assess the impact of free trade agreements on bilateral trade flows.
	
	\bigskip
	
\end{abstract}

\clearpage
	
\section{Introduction} \label{sec:introduction}

Dyadic data are ubiquitous in our society.
International trade, travels, population flows, military alliances, partnerships between firms, research collaboration, and many others can be represented as dyadic data, where each \textit{dyad} represents a pair of countries, firms, or individuals, depending on the context.
Dyadic data analysis is particularly prevalent in the literature of international trade, where regression-based analysis, the so-called \textit{gravity model}, serves as a primary analytical approach in these fields since the pioneering work by \cite{tinbergen1962shaping}.

Despite the popularity of dyadic data, there are only a few causal inference methods tailored specifically for dyadic data analysis, with some exceptions such as \cite{baier2009estimating}, \cite{arpino2017implementing}, and \cite{nagengast2025staggered}.
This may be due to the non-standard and complex endogeneity structure often encountered in typical applications of dyadic data.
For example, suppose that we are interested in the impacts of free trade agreements (FTAs) on trade flows between countries.
The treatment variable should be considered endogenous because both the decision to enter into an FTA and the trade outcome should be influenced by each country's economic factors and the economic and political relationship between the countries involved.
Thus, if researchers try to resolve the endogeneity issue by using the instrumental variables method, for instance, then they need to prepare at least two different types of instruments: those accounting for confounding factors at the ``origin'' country and those at the ``destination''.
This difficulty of finding appropriate instruments in the context of estimating the impacts of FTAs on trade flows is discussed in depth in \cite{baier2007free}.
In addition, not limited to international trade, dyadic data typically emerge in observational studies where it is difficult to find exogenous variations that help the researcher alleviate the endogeneity issue.

In this paper, we develop an easy-to-implement causal inference method specialized for dyadic data, where the treatment variable may potentially be endogenous, correlating with unobserved confounders in an arbitrary manner.
We suppose a situation where researchers have access to non-experimental data, comprising solely dyad-level treatments and outcomes.
We do not rely on any natural-experimental variations, such as instrumental variables, unexpected policy changes, policy-induced discontinuities, etc.

Our causal inference method is built on the following two requirements: (i) the dyadic treatments follow an exchangeable distribution, and (ii) only unit-level factors (rather than dyad-level factors) can cause endogeneity in the treatment choice.
The first condition is a common requirement in the literature of network statistics concerning the estimation of \textit{graphons}, where a graphon is a symmetric nonparametric function that generates graphs (i.e., networks) of any sizes.
As an important special case, if the treatments are independent and identically distributed (IID) across all dyads, they are automatically exchangeable.
Since requiring the exchangeability can be restrictive in some empirical situations, researchers should verify it on a case-by-case basis.
Nonetheless, it still encompasses a wide range of network formation models considered in the literature (cf. \citealp{graham2020network}).
Fully leveraging the exchangeability of treatments enables us to nonparametrically estimate the propensity score at each dyad without using any additional covariate information.
The second condition limits the sources of endogeneity to unit-level attributes, but it should be noted that variables representing (dis)similarities or social/economic distances between two units are not excluded as sources of endogeneity, as long as they can be expressed as functions of unit-level variables.

Under these assumptions, we study estimation and inference for several dyadic causal parameters.
The most basic component in our analysis is the \textit{dyadic} average treatment effect (ATE), defined as the difference in the conditional means of potential outcomes given the propensity score.
We show that the dyadic ATE can be estimated by regressing the outcome on the propensity score separately for the treatment and control groups.
Given this result, we propose a nonparametric kernel smoothing method for estimating the dyadic ATE and the propensity score, which we call \textit{neighbourhood kernel smoothing}.
The neighbourhood kernel smoothing method extends the neighbourhood smoothing approach introduced by \cite{zhang2017estimating}, which was originally developed for estimating probabilities of network edges.
Once the dyadic ATE is estimated for all dyads, we can then estimate the \textit{individual} ATE and the \textit{global} ATE.
The individual ATE represents the unit-level average treatment effect obtained by averaging the dyadic ATEs over the ``destinations'', while the global ATE is simply the average of dyadic ATEs over all dyads.
Under a set of regularity conditions, we derive the rate of convergence of the neighbourhood kernel smoothing estimator.
Furthermore, we develop conformal inference methods for constructing prediction sets for outcomes conditional on treatment status.

As an empirical application of our method, we investigate the impact of FTAs between two countries on export and import volumes.
Our dataset consists of 37 selected countries and regions in Asia, Oceania, and North America.
The outcome variables of interest are the share of export amount to the partner country among total exports to all countries and the share of import similarly defined, both for the year 2021.
For the treatment variable, we consider all FTAs related to our sample countries that were active as of 2021.
Our empirical results suggest that individual ATE of FTAs is positive for almost all countries (except Singapore) for both exports and imports, and that enactment of an FTA increases bilateral trade flows by approximately 2\% on average.

\textbf{Notation.}
Throughout the paper, $C$ (possibly with subscripts) denotes a global constant that does not depend on other variables, whose value may differ in different contexts.
We write a matrix whose $(i,j)$-th element is $A_{ij}$ as $A = (A_{ij})$ and its $i$-th row as $A_{i \cdot}$.
For random variables $X$ and $Y$, we write $X \overset{d}{=} Y$ if $X$ and $Y$ follow the same probability distribution, $X \lesssim Y$ if $X = O(Y)$ almost surely (a.s.), and $X \lesssim_{\bbP} Y$ if $X = O_{\bbP}(Y)$.

\section{Setup}\label{sec:setup} 

\subsection{Dyadic data}

Consider a population of interest composed of individuals, households, firms, or countries, depending on the context, that are interacting to each other pairwisely.
We observe $n$ units $[n] = \{ 1, 2, \dots, n \}$ drawn from the population, and we treat each pair of units (i.e., dyad) as one observation.
We are interested in estimating the causal effect of a dyadic treatment variable $A_{ij} \in \{ 0, 1 \}$ on a dyadic outcome variable $Y_{ij} \in \mathbb{R}$, where $(i, j)$ represents a dyad.
Throughout, we focus on the case of undirected treatment such that $A_{ij} = A_{ji}$ for all $(i, j)$.
As opposed to the treatment, the outcome variable can be specific to each unit in each dyad; that is, $Y_{ij} \neq Y_{ji}$ in general.
Let $A = (A_{ij})$ and $Y = (Y_{ij})$ denote the $n \times n$ observed matrices of treatments and outcomes, respectively, where the diagonals are normalized to zero: $A_{ii} = Y_{ii} = 0$ for all $i \in [n]$.

The above setup encompasses many empirical situations of interest.
For instance, in our empirical analysis, $A_{ij}$ represents the indicator for the presence of FTA between countries $i$ and $j$, and $Y_{ij}$ corresponds to the amount of exports from $i$ to $j$ or imports from $j$ to $i$.
In other examples, $A_{ij}$ may denote whether firms $i$ and $j$ participate in a researchers exchange program, and $Y_{ij}$ represents the R\&D collaboration output; $A_{ij}$ may be the indicator of political alliance between countries $i$ and $j$ with $Y_{ij}$ representing population migration; and so forth.

The dyadic data $\{(Y_{ij},A_{ij})\}_{i,j\in[n]}$ generally exhibit complex dependence structure due to interactions within and across dyads.
To account for such dependence, we consider a unit-level latent attribute $U_i$, which summarizes all observed and unobserved structural features of unit $i$ relevant to treatment assignment and potential outcomes.
Then, we suppose that the treatment assignment and potential outcomes are determined by
\begin{align}\label{eq:model}
	\begin{split}
		A_{ij} 
		& = \bm{1}\left\{ f(U_i,U_j) \ge U_{ij} \right\}, \\
		Y_{ij}(a) 
		& = y_a\left( U_i, U_j, \xi_{ij}(a) \right),
	\end{split}
\end{align}
for $i \in [n]$, $j \neq i$, and $a \in \{0, 1\}$.
Here, $f$ and $y_a$ are unknown functions, and $U_{ij}$ and $\xi_{ij}(a)$ are unobserved dyad-level attributes.
This specification introduces a dependence structure that is natural for dyadic data analysis.
For example, for dyads sharing unit $i$, the variables $A_{ij}$, $A_{ik}$, $Y_{ij}(a)$, and $Y_{ik}(a)$ are allowed to be systematically related through $U_i$.

The function $f$ in \eqref{eq:model} corresponds to what is known as a \textit{graphon} in network analysis.
A \textit{graphon} is a symmetric measurable function $f: [0,1]^2 \to [0,1]$, where $f(U_i, U_j)$ is interpreted as the edge probability between $i$ and $j$.
We say that an infinite matrix $G = (G_{ij})_{i,j \in \mathbb N}$ is \textit{exchangeable} if $(G_{ij}) \overset{d}{=} (G_{g(i)g(j)})$ for any permutation $g$ of $\mathbb N$.
Exchangeability plays an essential role in our analysis through the Aldous--Hoover representation theorem (\citealp{aldous1981representations}; \citealp{hoover1979relations}): an infinite symmetric random matrix $G$ with binary entries is exchangeable if and only if there exists a symmetric function $f:[0,1]^2 \to [0,1]$ such that for all $i < j$,
\begin{align*}
	G_{ij} = \bm{1}\{f(U_i,U_j)\ge U_{ij}\},
\end{align*}
where $\{U_i\}$ and $\{U_{ij}\}$ are mutually independent and $U_i, U_{ij} \overset{\mathrm{IID}}{\sim} \mathrm{Uniform}[0,1]$ with $U_{ij} = U_{ji}$.
This claim is equivalent to $G_{ij}\mid U_i,U_j \sim \mathrm{Bernoulli}(f(U_i,U_j))$ for all $i < j$ with a symmetric function $f$.
Thus, if the observed treatment matrix $A$ is an $n \times n$ submatrix of $G$, then the edge probabilities in $A$ admit the graphon representation $A_{ij} = \bm{1}\{f(U_i,U_j)\ge U_{ij}\}$.
Importantly, the specification that $U_i,U_{ij} \sim \mathrm{Uniform}[0,1]$ is only a normalization in this particular representation and does not preclude non-uniform, multivariate latent features in the underlying structural network formation model; see Lemma 2.22 of \cite{kallenberg1997foundations} for a related representation result.

Formally, we impose the following assumption.

\begin{assumption}[Treatment mechanism] \label{as:treatmentmat}
	For all $i \in [n]$ and $j \neq i$, $A_{ij}$ is determined as in \eqref{eq:model}, where $f: [0,1]^2 \to [0,1]$ is symmetric, $U_i,U_{ij} \overset{\mathrm{IID}}{\sim} \mathrm{Uniform}[0,1]$, and $U_{ij}=U_{ji}$.
\end{assumption}

While the IID-ness of $A_{ij}$ is an obvious sufficient condition for Assumption \ref{as:treatmentmat}, the assumption accommodates a more flexible class of network formation models.
This point is demonstrated by the following examples, in which such models arise naturally in empirical settings such as international trade, corporate governance, and R\&D collaboration; see \citet{graham2020network} for an overview. 

\begin{example}[Degree-heterogeneity models] \label{ex:heterogeneity}
	Suppose that the conditional link probability of each dyad is given by $P_{ij} = F(\pi(X_i, X_j) + \beta_i + \beta_j)$, where $F$ is a link function, $\pi$ is a symmetric function, $X_i$ is a vector of unit-level covariates, and $\beta_i$ represents a random degree heterogeneity effect.
	A similar model is considered in \cite{graham2017econometric}.
	If $\{(X_i, \beta_i)\}$ are IID, irrespective of the dimensionality or observability of $X_i$ and $\beta_i$, the resulting network is exchangeable.
\end{example}

\begin{example}[Stochastic block models with random membership] \label{ex:stochastic}
	Suppose that the population of units can be classified into $m$ groups, and that the link probability between $i$ and $j$ is given as $P_{ij} = \sum_{1 \le a, b \le m} z_{i,a} z_{j, b} p_{ab}$, where $p_{ab} \in (0,1)$ for all $1 \le a,b \le m$ such that $p_{ab} = p_{ba}$, and $z_{i,a} \ge 0$ is the membership variable satisfying $\sum_{a = 1}^m z_{i,a} = 1$ for all $i \in [n]$.
	This model can be viewed as a mixed-membership \textit{stochastic block model}.
	Here, assume that the membership variable is random and IID.
	This type of model has been studied widely in the literature (e.g., \citealp{airoldi2008mixed, white2016mixed}).
	Clearly, the resulting network is exchangeable.
\end{example}

\begin{example}[Pairwise strategic interaction models] \label{ex:pairwise}
	The exchangeability does not preclude strategic interaction, as long as it occurs within each dyad, as in \cite{hoshino2022pairwise}.
	This is empirically relevant since the treatments of interest in our model (e.g., FTAs, alliance, R\&D collaboration) are likely determined through cooperative decision-making.
	For example, we can consider a cooperative game model such that $P_{ij} = F(\pi(X_i, X_j) + \pi(X_j, X_i))$, where $\pi(X_i, X_j)$ denotes the payoff for $i$ from connecting to $j$.
	This model still fits into our framework if $\{X_i\}$ are IID.
\end{example}

The idea of using a graphon to model treatment assignment is not exclusive to our paper; it has also been explored by \citet{toulis2025estimating}.
They study unit-level treatments defined as functions of pre- and post-treatment networks, where edge probabilities are described by a graphon-type function of observed covariates.
Their primary focus is on the entanglement among unit-level treatments operating through the network.
By contrast, in our setting, the dyadic treatment itself forms the network $A$, and the individual entries of $A$ directly represent the treatment assignments of interest.

One limitation of Assumption \ref{as:treatmentmat} is that the resulting network will automatically be a dense network.
It is easy to see that the expected degree for each unit is a constant times $n-1$, which diverges as $n$ increases.
In order to accommodate the sparsity, \cite{klopp2017oracle} proposed the so-called sparse graphon estimation method.
Extending our method in this direction would be important, but is left for future study.

\subsection{Causal parameters}

Define the propensity score by
\begin{align*}
	P_{ij} \equiv \bbE[ A_{ij} \mid U_i, U_j ] = f(U_i, U_j).
\end{align*}
Note that $P_{ij}$ is a function of unit-level latent attributes $(U_i, U_j)$, differently from the standard propensity score defined with respect to observed covariates.
Also notice, for all $(i, j)$, that $P_{ij} = P_{ji}$ by Assumption \ref{as:treatmentmat} and that $P_{ij} = \bbE[ A_{ij} \mid P_{ij} ]$ by the law of iterated expectations.
Since $A_{ii} = 0$, we normalize $P_{ii} = 0$ for all $i \in [n]$.

By construction, $Y_{ij} = A_{ij} Y_{ij}(1) + (1 - A_{ij}) Y_{ij}(0)$.
The treatment effect of $A_{ij}$ at dyad $(i, j)$ is defined as $Y_{ij}(1) - Y_{ij}(0)$.
Similarly as above, we set $Y_{ii}(0) = Y_{ii}(1) = 0$ for all $i \in [n]$.
We here implicitly assume the stable unit treatment value assumption (SUTVA), which rules out treatment spillovers across dyads.
Importantly, this does not mean that treatments and outcomes are independent across dyads.
For example, in our framework, $U_i$ may affect both the treatment $A_{ij}$ for dyad $(i,j)$ and the potential outcome $Y_{ik}(a)$ for another dyad $(i,k)$.
Thus, dependence across dyads sharing a common unit is allowed through unit-level attributes.

What SUTVA rules out is the direct dependence of dyadic potential outcomes on the treatments of other dyads.
Once such spillovers are allowed, the potential outcome may depend on the treatment assignments of many other dyads, and potentially on the entire treatment configuration of the network.
Although this type of spillover may be empirically relevant in dyadic settings, without strong additional restrictions, it creates an enormous number of potential outcomes and renders identification and estimation intractable.
A tractable specification may be possible in special cases where spillovers arise only through a small number of known important partners; for details on such a model, see Appendix \ref{app:relax_sutva}.
\medskip

As the main causal parameter of interest, we define the \textit{dyadic} ATE as follows:
\begin{align*}
	\tau_{ij} 
	\equiv \bbE[ Y_{ij}(1) - Y_{ij}(0) \mid P_{ij} ]
	\qquad
	\text{for $i \in [n]$, $j \neq i$}. 
\end{align*}
Other causal parameters of interest would include the \textit{individual} ATE for unit $i$ and the \textit{global} ATE, which are respectively defined as
\begin{align*}
	\bar \tau_i^{\mathrm{IATE}}
	\equiv \frac{1}{n - 1} \sum_{j \neq i} \tau_{ij},
	\qquad 
	\bar \tau^{\mathrm{GATE}}
	\equiv \frac{1}{n} \sum_{i \in [n]} \bar \tau_i^{\mathrm{IATE}}.
\end{align*}
Since both individual and global ATEs can be estimated once the dyadic ATEs are estimated, in the following, we mainly focus on the dyadic ATE.
Since the dyadic ATE $\tau_{ij}$ is indexed by the propensity score $P_{ij}$, estimating $\tau_{ij}$ is empirically useful for understanding treatment effect heterogeneity across dyads.
Specifically, it can capture returns-to-selection patterns: if $\tau_{ij}$ tends to be larger for dyads with higher $P_{ij}$, it implies that those more likely to receive the treatment also tend to benefit more from it.
This interpretational advantage of conditioning on the propensity score has been similarly highlighted in recent studies (e.g., \citealp{zhou2019marginal}).

Meanwhile, since $P_{ij}$ can be a nontrivial function of $(U_i,U_j)$, it would also be natural to consider the ``finer'' counterpart as follows:
\begin{align*}
	\tau_{ij}^*
	\equiv
	\bbE\left[ Y_{ij}(1) - Y_{ij}(0) \mid U_i, U_j \right]
	\qquad
	\text{for $i \in [n]$, $j \neq i$}.
\end{align*}
Conceptually, $\tau_{ij}^*$ is closer in spirit to the marginal treatment effect parameter (e.g., \citealp{heckman2005structural}), as it is defined by conditioning on latent variables that affect treatment selection.
However, unlike the original marginal treatment effect framework where the conditioning latent variable acts as a measure of unobserved resistance to the treatment, $\tau_{ij}^*$ generally lacks a straightforward interpretation in our setting.
This is because the latent variable $U_i$ represents a composite factor whose role in treatment selection is unclear even with knowledge of underlying structural network model.
Nevertheless, the proof of the convergence rate below relies on an injectivity condition under which $\tau_{ij}$ and $\tau_{ij}^*$ coincide.
This condition is imposed only to derive the convergence rate of the proposed estimator, and is not required for the characterization of the dyadic ATE or the inference methods developed below.
See Assumption \ref{as:PO}(i) and Remark \ref{rem:alt} in the next section for more details.

\section{Estimation} \label{sec:method}

\subsection{Estimation procedure} \label{subsec:estimation}

We first introduce the following two assumptions.

\begin{assumption}[Selection on unit-level attributes] \label{as:independence1}
	For all $i \in [n]$, $j \neq i$, and $a \in \{0,1\}$, the potential outcome $Y_{ij}(a)$ is generated according to \eqref{eq:model}, where $\xi_{ij}(a)$ is conditionally independent of $U_{ij}$ given $(U_i, U_j)$.
\end{assumption}

\begin{assumption}[Overlap] \label{as:overlap}
	There exist $\underline{C}_P, \overline{C}_P \in (0, 1)$ such that $P_{ij} \in [\underline{C}_P, \overline{C}_P]$ for all $i \in [n]$ and $j \neq i$.
\end{assumption}

While Assumption \ref{as:overlap} is a standard overlap condition, Assumption \ref{as:independence1} should be carefully examined case by case.
Assumption \ref{as:independence1} implies that $Y_{ij}(a)$ is conditionally independent of $A_{ij}$ given $(U_i, U_j)$.
That is, only the unit-level latent factors $(U_i, U_j)$ can be the source of endogeneity, so that the potential outcomes and treatment are correlated only through $(U_i,U_j)$.
Assumption \ref{as:independence1} is violated when there are latent dyad-level factors that affect both the potential outcomes and the treatment.
For instance, in international trade applications, such factors may include historical ties, diplomatic relationships, or other bilateral affinities, which serve as confounders that are not reducible to independent country-level attributes.
In this case, addressing endogeneity would require alternative identification strategies, including instrumental variables.

Under these assumptions, we obtain the following result.
\begin{proposition}[Characterization of dyadic ATE] \label{prop:DATE}
	Under Assumptions \ref{as:treatmentmat} -- \ref{as:overlap}, for all $i \in[n]$, $j \neq i$, we have
	\begin{align*}
		\tau_{ij} = m_{ij}(1) - m_{ij}(0),
	\end{align*}
	where $m_{ij}(a) \equiv \bbE[ Y_{ij} \mid A_{ij} = a, P_{ij}]$ for $a \in \{ 0, 1 \}$.
\end{proposition}

Proposition \ref{prop:DATE} implies that we can estimate the dyadic ATE through a nonparametric regression of the outcome on the propensity score separately for the treatment and control groups.
However, since $\{P_{ij}\}$ are unknown, this is not feasible.
Alternatively, if we could identify a sub-sample satisfying $\{i' \in [n]: A_{i'j} = a, P_{i'j} \approx P_{ij}\}$, then computing the average of $Y_{i'j}$'s over this subset would give us a valid estimate of $m_{ij}(a)$.
We pursue this task by adopting the approach of \cite{zhang2017estimating}.
Let $d(i, i')$ denote the squared $\ell_2$ distance between graphon slices:
\begin{align}\label{eq:slice}
	d(i, i') \equiv \int_0^1 \left| f(U_i, v) - f(U_{i'}, v)\right|^2 \text{d}v.
\end{align}
Intuitively, if $d(i,i')$ is close to zero, we would have $P_{ij} \approx P_{i'j}$ for all $j \neq i,i'$.
Thus, if $d(i,i')$ were known, we can estimate $m_{ij}(a)$ by computing the average outcome over the subsample $\{i' \in [n]: A_{i'j} = a, d(i,i') \approx 0\}$.
Of course, directly estimating the value of 
$d(i,i')$ is not feasible.
However, for the purpose of selecting a neighbourhood of $i$, it is not necessary to know the precise value of $d(i,i')$; rather, it suffices to consider a tractable upper bound.
\cite{zhang2017estimating} demonstrated that the following pseudometric $\tilde d(i,i')$ serves as a viable approximate upper bound, in the sense that $\tilde d(i,i') \approx 0$ implies $d(i,i') \approx 0$ with a high probability:
\begin{align}\label{eq:dtilde}
	\tilde d(i, i')
	\equiv \max_{k \neq i, i'} \left| \frac{1}{n} \sum_{l \in [n]} (A_{il} - A_{i'l}) A_{kl} \right|.
\end{align}
Then, using this pseudometric in place of $d(i,i')$, we can consider a general kernel regression procedure, which we call the \textit{neighbourhood kernel smoothing} estimation.

The pseudometric $\tilde d$ is designed to identify units with similar graphon slices $f(U_i, \cdot)$, but not necessarily similar underlying latent variables $U_i$.
This is because, depending on the functional form of the graphon $f$, similar values of $P_{ij}$ can be generated even when the underlying pairs $(U_i, U_j)$ are substantially different.
Section 2.2 of \citet{zhang2017estimating} also highlights this point.

To build intuition for the pseudometric $\tilde d$, Figure \ref{fig:intuition} depicts a simple six-unit network, while Table \ref{tab:intuition} reports the corresponding values of $\tilde d(i,i')$.
In this example, units 1 and 2  have exactly the same connection pattern: both are connected to units 3, 4, and 5, and not connected to unit 6.
Accordingly, the pseudometric assigns $\tilde d(1, 2) = \tilde d(2, 1) = 0$.
By contrast, units with more distinct connection patterns, such as units 1 and 4, have a larger pseudometric distance, suggesting that their graphon slices should be substantially different.

\begin{figure}[t]
	\centering
	\begin{minipage}[c]{0.4\textwidth}
		\centering
		\begin{tikzpicture}[
			scale=0.7,
			transform shape,
			main node/.style={circle, draw=black, minimum size=1.5cm, font=\LARGE},
			edge/.style={thick}
			]
			\node[main node] (1) at (0, 3) {1};
			\node[main node] (3) at (3, 3) {3};
			\node[main node] (2) at (6, 3) {2};
			\node[main node] (4) at (0, 0) {4};
			\node[main node] (5) at (3, 0) {5};
			\node[main node] (6) at (6, 0) {6};
			\draw[edge] (1) -- (3);
			\draw[edge] (3) -- (2);
			\draw[edge] (1) -- (4);
			\draw[edge] (1) -- (5);
			\draw[edge] (3) -- (5);
			\draw[edge] (3) -- (6);
			\draw[edge] (2) -- (4);
			\draw[edge] (2) -- (5);
		\end{tikzpicture}
		\captionof{figure}{Example network}
		\label{fig:intuition}
	\end{minipage}
	\hfill
	\begin{minipage}[c]{0.55\textwidth}
		\centering
		\captionof{table}{Pairwise pseudometric values $\tilde d(i, i')$}
		\label{tab:intuition}
		\vspace{2mm}
		\small
		\begin{tabular}{c|rrrrrr}
			\diagbox{$i$}{$i'$} & \multicolumn{1}{c}{1} & \multicolumn{1}{c}{2} & \multicolumn{1}{c}{3} & \multicolumn{1}{c}{4} & \multicolumn{1}{c}{5} & \multicolumn{1}{c}{6} \\
			\hline
			1 & -- & 0 & 0.333 & 0.5 & 0.333 & 0.333 \\
			2 & 0 & -- & 0.333 & 0.5 & 0.333 & 0.333 \\
			3 & 0.333 & 0.333 & -- & 0.167 & 0.167 & 0.333 \\
			4 & 0.5 & 0.5 & 0.167 & -- & 0.167 & 0.333 \\
			5 & 0.333 & 0.333 & 0.167 & 0.167 & -- & 0.333 \\
			6 & 0.333 & 0.333 & 0.333 & 0.333 & 0.333 & -- \\
		\end{tabular}
	\end{minipage}
\end{figure}

The proposed neighbourhood kernel smoothing estimation proceeds as follows.
For a kernel weighting function $\mathcal{K}: \mathbb{R}_+ \to \mathbb{R}_+$ and a bandwidth $b > 0$, we compute the propensity score estimator as follows:
\begin{align} \label{eq:Ptilde}
	\tilde P_{ij}
	\equiv \sum_{i' \neq i} w(i, i') A_{i'j}
	\qquad \text{with} \quad 
	w(i, i') \equiv \left. \mathcal{K}\left( \frac{\tilde d(i, i')}{b} \right) \middle/ \sum_{l \neq i} \mathcal{K}\left( \frac{\tilde d(i, l)}{b} \right) \right..
\end{align}
However, this estimator suffers from two issues.
First, while we have assumed $P_{ij} = P_{ji}$ in Assumption \ref{as:treatmentmat}, $\tilde P_{ij} \neq \tilde P_{ji}$ in general.
Second, even under Assumption \ref{as:overlap}, $\tilde P_{ij}$ can be very close to zero or one in practice.
To address these issues, we consider a trimmed and symmetrized version of $\tilde P_{ij}$:
\begin{equation} \label{eq:PSest}
	\hat P_{ij}
	= \hat P_{ji}
	\equiv \frac{ \check P_{ij} + \check P_{ji} }{ 2 }
	\qquad \text{with} \quad 
	\check P_{ij} \equiv 
	\begin{cases}
		\kappa_L & \text{if $\tilde P_{ij} \le \kappa_L$,} \\
		\kappa_U & \text{if $\tilde P_{ij} \ge \kappa_U$,} \\
		\tilde P_{ij} & \text{otherwise,}
	\end{cases}
\end{equation}
where $0 < \kappa_L < \kappa_U < 1$ are user-specified trimming thresholds.
With this propensity score estimator, the neighbourhood kernel smoothing estimators for $m_{ij}(0)$ and $m_{ij}(1)$ are given by
\begin{equation} \label{eq:mest}
	\hat m_{ij}(0)
	\equiv \frac{ \sum_{i' \neq i} w(i, i') (1 - A_{i'j}) Y_{i'j} }{ 1 - \hat P_{ij} },
	\qquad 
	\hat m_{ij}(1)
	\equiv \frac{ \sum_{i' \neq i} w(i, i') A_{i'j} Y_{i'j} }{ \hat P_{ij} },
\end{equation}
respectively.
Finally, we estimate the dyadic ATE by
\begin{equation} \label{eq:DATEest}
	\hat \tau_{ij}
	\equiv \hat m_{ij}(1) - \hat m_{ij}(0).
\end{equation}
Moreover, the estimators for the individual and global ATE are obtained by
\begin{equation} \label{eq:IATEest}
	\hat \tau_i^{\mathrm{IATE}}
	\equiv \frac{1}{n - 1} \sum_{j \neq i} \hat \tau_{ij},
	\qquad 
	\hat \tau^{\mathrm{GATE}}
	\equiv \frac{1}{n} \sum_{i \in [n]} \hat \tau_i^{\mathrm{IATE}},
\end{equation}
respectively.

\begin{remark}[Related pseudometric-based approaches]
	\label{rem:applications}
	The pseudometric \eqref{eq:dtilde} has been used in several related settings beyond graphon estimation.
	In the context of regression models on networks, \citet{auerbach2022identification} and \citet{zeleneev2026identification} use closely related pseudometric measures to control for latent unit-level heterogeneity.
	In another context, for causal matrix completion based on panel data, \citet{athey2025identification} and \citet{deaner2025inferring} use the same type of pseudometric to estimate the ATE on the treated.
	In particular, our neighbourhood kernel smoothing estimator is closely related to the estimator for the conditional mean of the untreated potential outcome in \citet{deaner2025inferring}.
	They construct a pseudometric based on the pre-treatment outcome matrix generated by a non-linear factor model, whereas our approach applies it to the treatment matrix generated from the graphon $f$.
	How to incorporate pseudometric information for the outcome matrix $Y$, together with the pseudometric \eqref{eq:dtilde} for the treatment matrix $A$, into our setting should be an interesting topic for future research.
\end{remark}

\begin{remark}[Other estimation approaches]\label{rem:otherest}
	While this paper focuses on the inverse probability weighting estimators $\hat m_{ij}(0)$ and $\hat m_{ij}(1)$ in \eqref{eq:DATEest}, Proposition \ref{prop:DATE} also suggests that once the propensity score $P_{ij}$ is estimated via neighbourhood smoothing, the dyadic ATE could be estimated using alternative methods, such as propensity score matching and regression.
	We leave this promising direction for future research.
\end{remark}

\subsection{Rate of convergence} \label{subsec:convergence}

In this subsection, we derive the convergence rate of the neighbourhood kernel smoothing estimator.
We impose the following assumptions on the kernel weighting function $\mathcal{K}$, the bandwidth $b$, and the trimming thresholds $\kappa_L$ and $\kappa_U$.

\begin{assumption}[Kernel] \label{as:kernel}
	The kernel function $\mathcal{K}: \mathbb{R}_+ \to \mathbb{R}_+$ satisfies 
	(i) $\mathcal{K}(u) = 0$ for all $u > 1$, and
	(ii) there exist $\underline{C}_{\mathcal{K}}, \overline{C}_{\mathcal{K}} \in (0, \infty)$ such that $\mathcal{K}(u) \in [\underline{C}_{\mathcal{K}}, \overline{C}_{\mathcal{K}}]$ for all $u \in [0, 1]$.
\end{assumption}

\begin{assumption}[Bandwidth] \label{as:bw}
	Letting $h \equiv C_0 \sqrt{(\log n) / n}$ with some $C_0 \in (0, 1]$, the bandwidth $b$ is set to the $h$-th sample quantile of $\{ \tilde d(i, i'): i' \in [n], i' \ne i \}$.
\end{assumption}

\begin{assumption}[Trimming] \label{as:trim}
	The trimming thresholds $\kappa_L$ and $\kappa_U$ satisfy $0 < \kappa_L < \underline{C}_P < \overline{C}_P  < \kappa_U < 1$, where $\underline{C}_P$ and $\overline{C}_P$ are as in Assumption \ref{as:overlap}.
\end{assumption}

Assumption \ref{as:kernel} requires the kernel function $\mathcal{K}$ to be supported on the unit interval and bounded away from zero and above by a constant on the support.
The uniform kernel function is a typical example that satisfies this assumption.
Furthermore, commonly used compact support kernel functions can satisfy Assumption \ref{as:kernel} if we make slight modifications.
For instance, in line with Assumption \ref{as:kernel}, the triangular and Epanechnikov kernel functions can be respectively modified as 
\begin{align*}
	\mathcal{K}_{\mathrm{tri}}(u) 
	& \equiv (\underline{C}_{\mathcal{K}} + 1 - u) \bm{1}\{ u \le 1 \},\\
	\mathcal{K}_{\mathrm{epa}}(u) 
	& \equiv [\underline{C}_{\mathcal{K}} + 0.75 (1 - u^2)] \bm{1}\{ u \le 1 \}.
\end{align*}
The motivation behind this modification stems from the discrete nature of $\tilde d(i, i')$.
That is, due to the limited variation in the value of $\tilde d(i,i')$, particularly for small $n$, it is often the case that a certain number of observations who are exactly $b$ distant away from $i$ exist.
Consequently, when using standard kernel weight functions whose value degenerates to zero at the boundary, the resulting number of observations involved in the estimation can unintentionally be small.
We numerically check this issue in Appendix \ref{sec:simulation2}.

Assumptions \ref{as:bw} and \ref{as:trim} are technical conditions needed to derive the rate of convergence for our estimator.
Assumption \ref{as:bw} imposes the exact same condition as the one introduced in Theorem 1 of \cite{zhang2017estimating}.
In our numerical simulations, we confirm that the estimator performs reasonably well with $C_0 = 1$, $\kappa_L = 0.01$, and $\kappa_U = 0.99$.

Next, we assume that the graphon $f$ is a piecewise-Lipschitz function, precisely following Definition 2 and the assumption of Theorem 1 in \citet{zhang2017estimating}.

\begin{definition}[Piecewise-Lipschitz graphon family] \label{def:Lipschitz}
	For any $\delta, L > 0$, let $\mathcal{F}_{\delta; L}$ denote a family of piecewise-Lipschitz graphon functions $f:[0, 1]^2 \to [0, 1]$ such that (i) there exists an integer $K \ge 1$ and a sequence $0 = x_0 < \cdots < x_K = 1$ satisfying $\min_{0 \le k \le K - 1} (x_{k+1} - x_k) \ge \delta$, and (ii) both $| f(u_1, v) - f(u_2, v) | \le L | u_1 - u_2 |$ and $| f(u, v_1) - f(u, v_2) | \le L | v_1 - v_2 |$ hold for all $u, u_1, u_2 \in [x_{k_1}, x_{k_1 + 1}]$, $v, v_1, v_2 \in  [x_{k_2}, x_{k_2 + 1}]$, and $0 \le k_1, k_2 \le K - 1$.
\end{definition}

\begin{assumption}[Piecewise-Lipschitz graphon] \label{as:Lipschitzpropensity score}
	The graphon $f$ is an element of $\mathcal{F}_{\delta_n; L}$ satisfying $\delta_n / \sqrt{(\log n) / n} \to \infty$ for a global constant $L > 0$ and the number of pieces $K_n \ge 1$ such that $K_n \to \infty$ while $\min_{1 \le k \le K_n} |I_k| / \sqrt{(\log n) / n} \to \infty$, where $I_k = [x_{k-1}, x_k)$ for $1 \le k \le K_n-1$, and $I_{K_n} = [x_{K_n-1}, x_{K_n}]$
\end{assumption}

The next assumption imposes certain restrictions on the data generating process for the potential outcomes.

\begin{assumption}[Potential outcome] \label{as:PO}
	(i) For all $i \in [n]$, $j \neq i$, and $a \in \{0, 1\}$, 
	\begin{align*}
		\bbE[ Y_{ij}(a) \mid U_i, U_j ] = \bbE[ Y_{ij}(a) \mid P_{ij} ]
		\;\; \text{(a.s.)}
	\end{align*}
	(ii) There exists $C_Y > 0$ such that $|Y_{ij}(a)| \le C_Y$ for all $i \in [n]$, $j \neq i$, and $a \in \{ 0, 1 \}$.
	(iii) There exists a constant $L_Y > 0$ such that
	\begin{align*}
		| \bbE[ Y_{i'j}(a) \mid P_{i'j} ] - \bbE[ Y_{ij}(a) \mid P_{ij} ] | \le L_Y | P_{i'j} - P_{ij} | \;\; \text{(a.s.)}
	\end{align*}
	for all $i \in [n]$, $i' \in \mathcal{N}_i$, and $a \in \{ 0, 1 \}$, where $\mathcal{N}_i \equiv \{ i' \in[n]: i' \neq i, \tilde d(i, i') \le b \}$.
\end{assumption}

Assumption \ref{as:PO}(i) imposes an injectivity condition that the conditional mean $\bbE[ Y_{ij}(a) \mid U_i, U_j ]$ depends on $(U_i,U_j)$ only through $P_{ij}$.
This condition ensures that $\tau_{ij} = \tau_{ij}^*$, where $\tau_{ij}^*$ is defined in Section \ref{sec:setup}.
While this assumption restricts the dependence structure of potential outcomes, it does not impose additional restrictions on treatment network formation.
We use this condition to establish Theorem \ref{thm:convergence} below.
It might be possible to obtain the same result in the theorem under weaker alternative conditions, however, relaxing this assumption is left for future work; see the proof of Lemma \ref{lem:J12} in Appendix \ref{sec:proof} for details.
In Assumption \ref{as:PO}(ii), we assume that the potential outcomes are bounded.
Although this is stronger than necessary, it significantly simplifies our theoretical analysis.
Assumption \ref{as:PO}(iii) is a high-level condition.
This can be satisfied, for instance, when $\{ Y_{ij}(a) \}$ are identically distributed such that $\bbE[ Y_{ij}(a) \mid P_{ij} ] = g_a(P_{ij})$, where $g_a$ is $L_Y$-Lipschitz on $[0, 1]$.

\begin{assumption}[Conditional independence] \label{as:independence2}
	For all mutually distinct $i, j, k \in [n]$ and all $a \in \{ 0, 1 \}$,
	(i) $(U_{ij}, \xi_{ij}(a))$ is independent of $(U_k, U_{kj}, \xi_{kj}(a))$ conditional on $(U_i, U_j)$, and
	(ii) $(U_j, U_{ij}, U_{kj}, \xi_{ij}(a), \xi_{kj}(a))$ satisfies mutual independence across $j$ conditional on $(U_i, U_k)$. 
\end{assumption}

Assumption \ref{as:independence2} imposes mild technical independence conditions, which are introduced to facilitate the application of Hoeffding's inequality in the proof of Lemma \ref{lem:J12}.
This assumption strengthens the independence conditions in Assumptions \ref{as:treatmentmat} and \ref{as:independence1}.
For instance, it holds if $\{ U_i \}$, $\{ U_{ij} \}$, and $\{ \xi_{ij}(a) \}$ are independent of each other, and furthermore, both $\{ U_{ij} \}$ and $\{ \xi_{ij}(a) \}$ are independent across $j$.

The next theorem gives the convergence rate of the neighbourhood kernel smoothing estimator for the dyadic ATE with respect to the the normalized $(2, \infty)$ matrix norm.
For $n \times n$ matrices $A$ and $B$ with zero diagonals, we define
\begin{align*}
	d_{2,\infty}(A, B)
	& \equiv \max_{1 \le i \le n} \| A_{i \cdot} - B_{i \cdot} \|_2 / \sqrt{n - 1},
\end{align*}
where $\| \cdot \|_2$ denotes the Euclidean norm.

\begin{theorem}[Rate of convergence] \label{thm:convergence}
	Under Assumptions \ref{as:treatmentmat} -- \ref{as:independence2}, 
	\begin{align*}
		\big[ d_{2,\infty}(\hat \tau, \tau) \big]^2 
		\lesssim_{\bbP} \sqrt{ (\log n) / n },
	\end{align*}
	where $\hat \tau = (\hat \tau_{ij})$ and $\tau = (\tau_{ij})$.
\end{theorem}

From Theorem \ref{thm:convergence} and Jensen's inequality, we can obtain a bound for the convergence rate of the individual ATE estimator as follows:
\begin{equation}\nonumber
	\frac{1}{n}\sum_{i = 1}^n \left|\hat \tau_i^{\mathrm{IATE}} - \tau_i^{\mathrm{IATE}}\right|^2
	\lesssim_{\bbP} \sqrt{ (\log n) / n }.
\end{equation}
Similarly, we can easily observe that $|\hat \tau^\text{GATE} - \tau^\text{GATE}|^2 \lesssim_{\bbP} \sqrt{ (\log n) / n }$ holds.

\begin{remark}[Alternative propensity score estimators]\label{rem:alt}
	It is straightforward to see that the same bound as in Theorem \ref{thm:convergence} applies to $\|\hat{\tau}-\tau\|_F^2 / [n(n-1)]$, where $\|\cdot\|_F$ denotes the Frobenius norm.
	Under certain smoothness conditions on graphons, the minimax rate for graphon estimation under this norm is known to be $(\log n)/n$;
	see \cite{gao2015rate}, \cite{zhang2017estimating}, and \cite{gao2021minimax}.
	Hence, our bound is not optimal from that perspective.
	One example of a graphon estimator that achieves the optimal rate is the combinatorial least squares estimator of \citet{gao2015rate}, but it is often computationally infeasible in practice.
	Another example is the sorting-and-smoothing estimator of \citet{chan2014consistent}, which relies on an additional monotonicity condition.
	By contrast, the neighbourhood smoothing approach of \citet{zhang2017estimating} is particularly attractive because it does not aim to identify the graphon $f$ itself or the underlying latent variables $\{ U_i \}$, thereby avoiding such stringent assumptions.
	This flexibility is crucial in our setting, as it accommodates empirically plausible treatment network formation models, such as those illustrated in Examples \ref{ex:heterogeneity} -- \ref{ex:pairwise}.
	In such settings, where the first step estimates the dyadic propensity score rather than the graphon $f$ itself or the underlying latent variables $\{ U_i \}$, the most natural causal parameter would be our dyadic ATE $\tau_{ij}$.
\end{remark}

\section{Predictive Inference}\label{sec:inference}

The neighbourhood kernel smoothing estimators are easy to implement and achieve certain rates of convergence, but deriving their asymptotic distributions is highly challenging.
This difficulty arises from the complex dependence structure induced by the pseudometric $\tilde d$, which is defined using the entire network $A$.
Consequently, statistical inference for neighbourhood smoothing estimators remains an open question in the literature.

Accordingly, rather than conducting statistical inference directly on our estimators, we consider employing a conformal inference approach to construct a prediction set for a new outcome $Y_{i,n+1}$ conditional on $A_{i,n+1}=a$, for each $i \in [n]$ and $a \in \{0,1\}$, in a similar way to \cite{lei2018distribution}.
Formally, we aim to construct a set $\mathcal{C}_{1-\alpha}(i, a)$ for each $i \in [n]$ and $a \in \{ 0, 1 \}$, such that
\begin{align*}
	\bbP_{ia} \big( Y_{i,n+1}(a) \in \mathcal{C}_{1-\alpha}(i, a) \big) \ge 1 - \alpha,
\end{align*}
where $\alpha \in (0,1)$ is a nominal miscoverage level and $\bbP_{ia}$ denotes the probability law of $\{(Y_{ij},A_{ij}): j \in [n+1], j \neq i\}$ conditional on $U_i$ and $A_{i,n+1}=a$.

The resulting prediction sets $\mathcal{C}_{1-\alpha}(i,0)$ and $\mathcal{C}_{1-\alpha}(i,1)$ may reflect self-selection bias and generally do not have a pure causal interpretation, since they are prediction sets for $Y_{i,n+1}(a)$ conditional on $A_{i,n+1}=a$, rather than prediction sets for the potential outcome $Y_{i,n+1}(a)$ itself.
However, comparing them can still provide supportive evidence for the obtained treatment effect estimates.
For instance, when positive self-selection is suspected, if $\mathcal{C}_{1-\alpha}(i,0)$ and $\mathcal{C}_{1-\alpha}(i,1)$ overlap to a large extent, this may be viewed as supportive evidence that the treatment effect is weak or negligible.

The key requirement to implement a conformal inference is that the data are exchangeable.
In our setting, since $A_{ij} = \bm{1}\{ f(U_i, U_j) \ge U_{ij} \}$ and $Y_{ij}(a) = y_a(U_i, U_j, \xi_{ij}(a))$ where $U_i, U_{ij} \overset{\text{IID}}{\sim} \text{Uniform}(0, 1)$, if $\xi_{ij}(a)$'s are also IID, then the $i$-th row $\{ (Y_{ij}, A_{ij}): j \neq i \}$ becomes exchangeable over $j$ conditional on $U_i$.
This justifies the finite-sample validity of a row-wise conformal inference.

Our row-wise full conformal inference proceeds as follows.
For each $i \in [n]$ and $a \in \{ 0, 1 \}$, select a subsample $\{ Y_{ij}: j \in \mathcal{J}(i, a) \}$, where $\mathcal{J}(i, a) \equiv \{ j \neq i : A_{ij} = a \}$, and let $y \in \mathbb{R}$ be a generic candidate value for $Y_{i,n+1}(a)$.
We then compute the following conformity scores:
\begin{align*}
	R_{y,j}(i, a)
	& \equiv \left| Y_{ij} - \frac{1}{|\mathcal{J}(i, a)| + 1} \left( \sum_{j \in \mathcal{J}(i, a)} Y_{ij} + y \right)  \right|
	\qquad \text{for $j \in \mathcal{J}(i, a)$,}
	\\
	R_{y,n+1}(i, a) 
	& \equiv \left| y - \frac{1}{|\mathcal{J}(i, a)| + 1} \left( \sum_{j \in \mathcal{J}(i, a)} Y_{ij} + y \right) \right|.
\end{align*}
Then, the conformal prediction set for $Y_{i,n+1}(a)$ conditional on $A_{i,n+1} = a$ can be obtained by
\begin{align*}
	\mathcal{C}_{1-\alpha}^{\mathrm{full}}(i, a)
	\equiv \{ y \in \mathbb{R} : (|\mathcal{J}(i, a)| + 1) \pi_y(i, a) \le \lceil (1 - \alpha)(|\mathcal{J}(i, a)| + 1) \rceil \},
\end{align*}
where $\lceil \cdot \rceil$ is the ceiling function and
\begin{align*}
	\pi_y(i,a) 
	\equiv \frac{1}{|\mathcal{J}(i, a)| + 1} \left(1 + \sum_{j \in \mathcal{J}(i, a)} \bm{1} \{ R_{y,j}(i, a) \le R_{y,n+1}(i, a) \} \right).
\end{align*}

Alternatively, we can consider split conformal inference as in \cite{lei2018distribution}, which is computationally more efficient than full conformal inference, particularly when $n$ is large.
For each $i \in [n]$ and $a \in \{0,1\}$, we randomly split $\mathcal{J}(i,a)$ into a training set $\mathcal{J}_{\mathrm{tr}}(i,a)$ and a calibration set $\mathcal{J}_{\mathrm{ca}}(i,a)$.
We use the training set $\mathcal{J}_{\mathrm{tr}}(i) \equiv \mathcal{J}_{\mathrm{tr}}(i,0) \cup \mathcal{J}_{\mathrm{tr}}(i,1)$ to obtain the ordinary least squares estimates $\hat \mu(i,0)=\hat \beta_i$ and $\hat \mu(i,1)=\hat \beta_i+\hat \tau_i$ in the regression model $Y_{ij}=\beta_i+A_{ij}\tau_i+\xi_{ij}$ for $j \in \mathcal{J}_{\mathrm{tr}}(i)$.
We then use the calibration set to compute the conformity scores:
\begin{align*}
	R_j(i,a) \equiv |Y_{ij}-\hat \mu(i,a)|
	\qquad 
	\text{for $j \in \mathcal{J}_{\mathrm{ca}}(i,a)$.}
\end{align*}
Then, the split conformal prediction interval for $Y_{i,n+1}(a)$ conditional on $A_{i,n+1}=a$ can be obtained by
\begin{align*}
	\mathcal{C}_{1-\alpha}^{\mathrm{split}}(i,a)
	\equiv \left[\hat \mu(i,a)-\hat q_{1-\alpha}(i,a), \; \hat \mu(i,a)+\hat q_{1-\alpha}(i,a)\right],
\end{align*}
where $\hat q_{1-\alpha}(i,a)$ is the $\lceil (1-\alpha)(|\mathcal{J}_{\mathrm{ca}}(i,a)|+1) \rceil$-th smallest value in $\{R_j(i,a): j \in \mathcal{J}_{\mathrm{ca}}(i,a)\}$.

Given the prediction intervals for $Y_{i,n+1}(1) \mid \{A_{i,n+1}=1\}$ and $Y_{i,n+1}(0) \mid \{A_{i,n+1}=0\}$, we can construct a Bonferroni-type prediction interval for their difference.
Specifically, let $[L_{1-\alpha/2}(i,a),U_{1-\alpha/2}(i,a)]$ denote the resulting prediction interval at level $1-\alpha/2$ for $Y_{i,n+1}(a) \mid \{A_{i,n+1}=a\}$ for $a \in \{0,1\}$.
Then, a valid unit-$i$ specific prediction interval for the outcome difference at level $1-\alpha$ can be obtained by
\begin{align*}
	\mathcal{C}_{1-\alpha}(i)
	\equiv \left[L_{1-\alpha/2}(i,1)-U_{1-\alpha/2}(i,0), \; U_{1-\alpha/2}(i,1)-L_{1-\alpha/2}(i,0)\right].
\end{align*}
However, this type of construction may be unnecessarily conservative (cf. \citealp{lei2021conformal}).
More importantly, this prediction interval for the conditional outcome difference does not generally admit a causal interpretation due to treatment endogeneity.

\begin{remark}[Permutation or randomization tests] \label{remark:permutation}
	Under additional assumptions, one can also consider permutation inference.
	For example, if we additionally assume independence between $Y(0)$ and $A$, then we can perform a permutation test based on the exchangeability in Assumption \ref{as:treatmentmat} to test the sharp null hypothesis of no treatment effect, $Y_{ij}(0) = Y_{ij}(1)$ for all $i \in [n]$ and $j \neq i$.
	We formalize this testing approach and establish its validity in Appendix \ref{sec:permutation}.
	Alternatively, when the graphon is parametrically specified as a function of observed covariates, one could construct conditional randomization tests in the spirit of \citet{toulis2025estimating}.
	This approach, however, would deviate from the main purpose of our flexible graphon framework.
	We therefore do not pursue this direction here.
\end{remark}

\section{Numerical Simulations}\label{sec:simulation1}

In this section, we conduct Monte Carlo simulations to evaluate the finite-sample performance of our estimation and predictive inference methods.
We consider two sample sizes, $n \in \{ 40, 80 \}$, and perform 1,000 replications for each data generating process.
For brevity, this section contains only the results from the main simulation setup, while supplementary results are deferred to Appendix \ref{sec:simulation2}.

We consider two data generating processes for propensity scores.
The first is a stochastic block model with a covariate.
Letting $X_{1i} \overset{\mathrm{IID}}{\sim} \mathrm{Uniform}[0, 1]$ be a unit-level covariate, the units are classified into three groups according to the following group assignment probability: 
\begin{align*}
	\text{pr}(Z_i = z \mid X_{1i}) = 
	\begin{cases}
		z X_{1i} / 3 & \text{if $z \in \{ 1, 2 \}$,} \\
		1 - X_{1i}   & \text{if $z = 3$,} 
	\end{cases}
\end{align*}
where $Z_i \in \{ 1, 2, 3 \}$ indicates the group to which agent $i$ is assigned.
We then generate the propensity score at each dyad as follows:
\begin{align*}
	P_{ij} = 
	\begin{cases}
		[1 + (1 + z) (X_{1i} + X_{1j})] / 10 & \text{if $Z_i = Z_j = z$ for $z \in \{ 1, 2, 3 \}$,} \\
		(1 + X_{1i} + X_{1j}) / 10 & \text{otherwise.}
	\end{cases} 
\end{align*}

For the second data generating process of the propensity score, we adopt the same network formation model as Design A.1 in \cite{graham2017econometric}:
\begin{align*}
	P_{ij}
	= F( X_{2i} X_{2j} + \beta_i + \beta_j ),
\end{align*}
where $F$ denotes the standard logistic cumulative distribution function, $X_{2i} \in \{ -1, 1 \}$ is a unit-level covariate such that $\text{pr}(X_{2i} = -1) = \text{pr}(X_{2i} = 1) = 0.5$, and
$\beta_i$ represents an individual-level degree heterogeneity defined as $\beta_i = -0.5 + V_i$ with a centered Beta random variable $V_i \overset{\mathrm{IID}}{\sim} ( \mathrm{Beta}(1, 1) - 0.5 )$.

The potential outcomes are generated by
\begin{align*}
	Y_{ij}(0) = \xi_{ij},
	\qquad 
	Y_{ij}(1) = Y_{ij}(0) + \gamma_0 \zeta_{ij} + \gamma_1 P_{ij} + \gamma_2 P_{ij}^2,
\end{align*}
where $\xi_{ij}, \zeta_{ij} \overset{\mathrm{IID}}{\sim} \mathrm{Normal}(0, 1)$ independent of $\{(X_{1i}, X_{2i}, V_i)\}$.
We consider the following three patterns for the values of $\gamma$'s:
(i) $\gamma_0 = \gamma_1 = \gamma_2 = 0$, (ii) $\gamma_0 = \gamma_1 = 1$, $\gamma_2 = 0$, and (iii) $\gamma_0 = \gamma_1 = \gamma_2 = 1$.
Under these three setups, the dyadic ATE parameter becomes (i) $\tau_{ij} = 0$, (ii) $\tau_{ij} = P_{ij}$, and (iii) $\tau_{ij} = P_{ij} + P_{ij}^2$, respectively.
Also notice that $Y_{ij}(0) = Y_{ij}(1)$ in case (i).
Hence, we refer to these as the null setup and the linear and quadratic dyadic ATE setups, respectively.

We specify the implementation details for the estimation and predictive inference methods as follows.
For the kernel weighting function, in view of Assumption \ref{as:kernel}, we consider three kernel functions with $\underline{C}_{\mathcal{K}} = 1$: (i) uniform kernel, (ii) triangular kernel, and (iii) Epanechnikov kernel.
In line with Assumption \ref{as:bw}, the bandwidth $b$ is set to the $h = \sqrt{ (\log n) / n }$-th quantile of $\{ \tilde d(i, i') \}_{i' \ne i}$.
The trimming thresholds for estimating the propensity scores are set to $(\kappa_L, \kappa_U) = (0.01, 0.99)$.
For full conformal inference, we consider $\alpha = 0.1$ as the nominal miscoverage level and a grid of 1,001 equally spaced points over the interval $[-5, 5]$ as the candidate outcome values $y$.

To assess the estimation of the dyadic, individual, and global ATEs, as well as the propensity score, we report the average bias and the average mean squared error:
\begin{align*}
	\text{AvgBias} & = \frac{1}{1000}\sum_{r = 1}^{1000}\left[ \frac{1}{n(n - 1)} \sum_{i \in [n]} \sum_{j \neq i} (\hat \tau_{ij}^{(r)} - \tau_{ij})\right], \\
	\text{AvgMSE} & = \frac{1}{1000}\sum_{r = 1}^{1000}\left[ \frac{1}{n(n - 1)} \sum_{i \in [n]} \sum_{j \neq i} (\hat \tau_{ij}^{(r)} - \tau_{ij})^2 \right],
\end{align*}
where the superscript $(r)$ indicates that it is obtained from the $r$-th replicated dataset.
Analogous definitions apply to the other parameters.
To evaluate the prediction interval $\mathcal{C}_{1-\alpha}(i, a) = [L_{1-\alpha}(i, a), \; U_{1-\alpha}(i, a) ]$ for $Y_{i,n+1}(a) | \{A_{i,n+1} = a\}$, we compute the average empirical coverage probability and the average length:
\begin{align*}
	\text{AvgCoverage}(a) & = \frac{1}{1000} \sum_{r=1}^{1000} \left[ \frac{1}{n} \sum_{i \in [n]} \bm{1}\{ A_{i,n+1}^{(r)} = a \} \bm{1}\{ L_{1-\alpha}^{(r)}(i, a) \le Y_{i,n+1}^{(r)}(a) \le U_{1-\alpha}^{(r)}(i, a) \}  \right], \\
	\text{AvgLength}(a) & = \frac{1}{1000} \sum_{r=1}^{1000} \left[ \frac{1}{n} \sum_{i \in [n]} \left( U_{1-\alpha}^{(r)}(i, a) - L_{1-\alpha}^{(r)}(i, a) \right) \right].
\end{align*}

Table \ref{tab:MC1-DGP2} reports the Monte Carlo simulation results for the neighbourhood kernel smoothing estimator under the linear dyadic ATE setup.
From this, we can observe that the average bias and the average mean squared error for the dyadic ATE and propensity score estimates are satisfactorily small, especially when $n$ is relatively large.
Furthermore, because the individual and global ATEs are obtained by aggregating the dyadic ATEs, the average mean squared errors for estimating these parameters are significantly smaller than that for the dyadic ATE.
Regarding the kernel weighting functions, all three kernels demonstrate similar performance; notably, the choice of kernel has almost no discernible impact on propensity score estimation.
For the dyadic ATE, however, the triangular and Epanechnikov kernels tend to yield slightly smaller biases, while the uniform kernel outperforms the others in terms of average mean squared error.
The same pattern is observed in the other DGPs.
Full simulation results are reported in Table \ref{tab:MC1} in Appendix \ref{sec:simulation2}.

For the proposed conformal inference method, we report the simulation results in Table \ref{tab:conformal} in Appendix \ref{sec:simulation2} to save space.
They show that the average empirical coverage probability is satisfactorily close to the nominal coverage level of 90\%, and that the average interval length is reasonable and tends to decrease as $n$ increases.

\begin{table}[t]
	\caption{Simulation results for neighbourhood kernel smoothing: Linear dyadic ATE setup} \label{tab:MC1-DGP2}
	\begin{center}
		\begin{small}
			\begin{tabular}[t]{rrrrrrrrrrrr}
				\multicolumn{3}{c}{DGP} & \multicolumn{2}{c}{DATE} & \multicolumn{2}{c}{IATE} & \multicolumn{2}{c}{GATE} & \multicolumn{2}{c}{PS} \\
				$P$ & $n$ & $\mathcal{K}$ & AvgBias & AvgMSE & AvgBias & AvgMSE & AvgBias & AvgMSE & AvgBias & AvgMSE \\
				1 & 40 & 1 & 3.65 & 63.68 & 3.65 & 2.30 & 3.65 & 0.86 & -1.73 & 1.76\\
				&  & 2 & 3.60 & 65.34 & 3.60 & 2.32 & 3.60 & 0.86 & -1.74 & 1.80\\
				&  & 3 & 3.58 & 66.15 & 3.58 & 2.33 & 3.58 & 0.86 & -1.74 & 1.81\\
				& 80 & 1 & 2.39 & 40.90 & 2.39 & 0.85 & 2.39 & 0.23 & -1.23 & 1.16\\
				&  & 2 & 2.32 & 41.77 & 2.32 & 0.85 & 2.32 & 0.23 & -1.22 & 1.16\\
				&  & 3 & 2.30 & 42.27 & 2.30 & 0.85 & 2.30 & 0.23 & -1.22 & 1.16\\
				2 & 40 & 1 & 4.72 & 82.47 & 4.72 & 3.13 & 4.72 & 0.94 & -1.75 & 1.46\\
				&  & 2 & 4.33 & 84.12 & 4.33 & 3.10 & 4.33 & 0.91 & -1.76 & 1.47\\
				&  & 3 & 4.26 & 84.85 & 4.26 & 3.10 & 4.26 & 0.91 & -1.76 & 1.47\\
				& 80 & 1 & 1.17 & 60.59 & 1.17 & 1.21 & 1.17 & 0.22 & -0.94 & 0.68\\
				&  & 2 & 1.04 & 61.87 & 1.04 & 1.19 & 1.04 & 0.22 & -0.95 & 0.71\\
				&  & 3 & 1.02 & 62.45 & 1.02 & 1.19 & 1.02 & 0.22 & -0.95 & 0.72\\
			\end{tabular}
		\end{small}
	\end{center}
	\begin{small}
		\noindent Note: All entries are multiplied by $100$.
		
		Abbreviations: DATE (dyadic ATE), IATE (individual ATE), GATE (global ATE), PS (propensity score).
		
		$\text{DGP }(P) \in \{\text{1 (stochastic block model)}, \text{2 (degree heterogeneity model)}\}$.
		
		$\mathcal{K} \in \{\text{1 (uniform)}, \text{2 (triangular)}, \text{3 (Epanechnikov)}\}$, with $\underline{C}_{\mathcal{K}}=1$.
	\end{small}
\end{table}

\section{Empirical Analysis} \label{sec:empiric}

As an empirical illustration of our method, we investigate the effects of FTAs on bilateral trade flows. 
This has been a classical and one of the most important empirical questions in international economics (e.g., \citealp{tinbergen1962shaping}).
In the literature, there seems to be an agreement that FTAs indeed help increase exports and imports between member countries; however, it is still somewhat unclear to what extent they are actually beneficial; see, for example, the discussion in \cite{nagengast2025staggered}.
In addition, as mentioned above, since most papers on this topic use a regression-based gravity model approach, which is not necessarily based on a causal inference framework, our study would contribute to the literature in this respect.

In this empirical analysis, we focus on the trade flows across 37 selected countries and regions mainly in Asia, Oceania, and North America.
The outcomes of interest are bilateral export and import amounts.
Considering huge variation and unbalance of these values across countries, we transform them into percents of total exports and imports with all countries.
The data of exports and imports for the year 2021 are obtained from the World Integrated Trade Solution, World Bank (\url{https://wits.worldbank.org}).

For the treatment variable, we consider all FTAs related to our sample countries that were enacted as of 2021.
The data for the FTAs are obtained from the website of Japan External Trade Organization (\url{https://www.jetro.go.jp/theme/wto-fta/ftalist.html}).
We do not differentiate between the contents of FTAs.
If a pair of countries jointly participates in at least one FTA, we consider that dyad as treated.

Figure \ref{fig:fta} shows the FTA network of our data.
Among the total of 666 dyads, the number of treated dyads is 197.
In our dataset, Singapore has the largest number of FTA partners (degree 27), and the country with the smallest is Mongolia (degree 1).
Based on this dataset, we estimate the treatment effects of FTAs using the same estimation procedure as in the simulation experiment in Section \ref{sec:simulation1}.
For the choice of kernel function, we use the uniform kernel.

\begin{figure}[t]
	\begin{center}
		\includegraphics[width = 15cm, bb = 0 0 814 664]{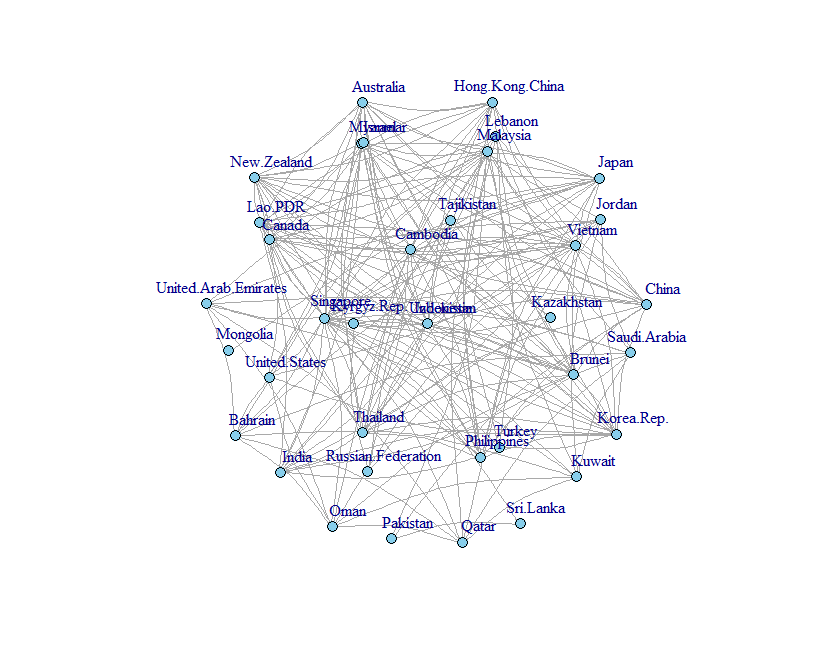}
		\caption{FTA network}
		\label{fig:fta}
	\end{center}
\end{figure}

Before presenting our empirical results, we recall two key assumptions underlying our analysis: the exchangeability of the FTA network and SUTVA.
As discussed in Sections \ref{sec:setup} and \ref{sec:method}, our framework can flexibly accommodate country-level heterogeneity, including geographic location, economic size, trade openness, and other unobserved factors.
Moreover, it allows the FTA status of dyad $(i,j)$ and the trade flow of another dyad $(i,k)$ to co-vary through the shared latent attribute $U_i$.
Thus, the framework does not entirely rule out cross-dyad dependence, providing a reasonable and tractable approximation for this application.
At the same time, however, these assumptions are not innocuous.
The FTA network may depend on dyad-specific factors, such as diplomatic relations outside of trade agreements.
In addition, an FTA between two countries may directly affect trade flows with third countries, for example, through trade diversion, which could violate the no-interference condition.
Consequently, the empirical results presented below should be interpreted with these limitations in mind.

The estimated individual ATEs on exports and imports for all countries are depicted in Figure \ref{fig:iate}.
We can observe that for both exports and imports, the impacts of FTAs are positive for almost all countries.
It appears that there is a positive correlation between the impacts on exports and those on imports, suggesting that FTAs are mutually beneficial for both ``origin'' and ``destination'' countries.
For both exports and imports, the presence of an FTA increases the trade volume between countries by up to approximately 4 percent points.
Averaged over the countries, the estimated global ATE on exports is 1.938 and that on imports is 2.161, suggesting that the impact of establishing an FTA, on average, increases trade inflow and outflow by about 2 percent of the total.
Meanwhile, we could not observe the impact of FTA only in Singapore.
As mentioned above, Singapore has FTAs with almost three-quarters of the countries in the dataset, which might have made it more difficult to identify the impact of FTAs compared with other countries.

To investigate the nature of treatment heterogeneity, we draw scatter plots of individual ATEs against individual average propensity scores in Figure \ref{fig:scatplot}.
Interestingly, the scatter plots reveal a certain negative correlation between the treatment effect and propensity score.
This observation might indicate the possibility of decreasing marginal returns to the number of FTAs. 

Lastly, Figure \ref{fig:ci} reports the country-wise full conformal prediction intervals for export and import shares conditional on the presence or absence of an FTA.
Since the effective sample size for each country is limited, the lower bound of the prediction intervals are  zero in most cases.
Nevertheless, we can observe a clear pattern that the upper bounds under $A = 1$ tend to be substantially larger than those under $A = 0$ for both exports and imports.
This finding is consistent with the positive estimated individual ATEs in Figure \ref{fig:iate}.
In particular, Mongolia, which exhibits a relatively large individual ATE, also has a large upper bound of the prediction interval under $A = 1$.
However, this comparison does not by itself have a causal interpretation, because the difference between the two prediction intervals reflects self-selection into FTAs as well as the treatment effect.
Thus, these prediction results should not be overstated.

\begin{figure}[t]
	\begin{minipage}[b]{\textwidth}
		\centering
		\includegraphics[width=0.8\textwidth, bb = 0 0 659 246]{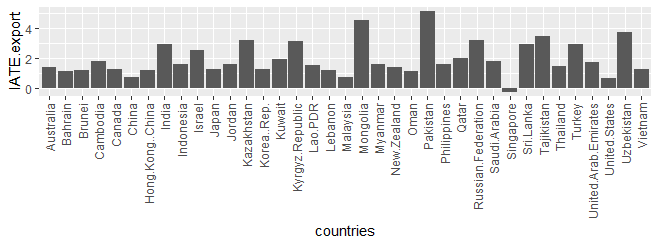}
		\subcaption{Export share (\%)} 
	\end{minipage}
	\begin{minipage}[b]{\textwidth}
		\centering
		\includegraphics[width=0.8\textwidth, bb = 0 0 659 246]{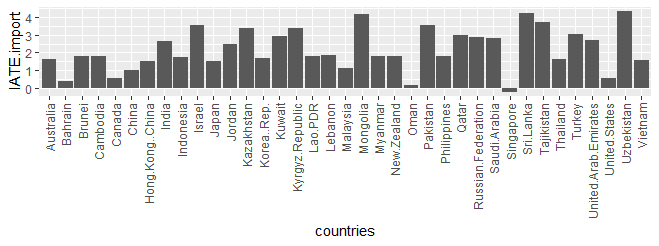}
		\subcaption{Import share (\%)}
	\end{minipage}
	\caption{The estimated individual ATEs} \label{fig:iate}
\end{figure}

\begin{figure}[t]
	\begin{minipage}[b]{0.5\textwidth}
		\centering
		\includegraphics[width=0.8\textwidth, bb = 0 0 418 280]{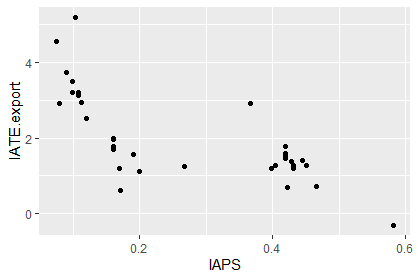}
		\subcaption{Export share (\%)}
	\end{minipage}
	\begin{minipage}[b]{0.5\textwidth}
		\centering
		\includegraphics[width=0.8\textwidth, bb = 0 0 418 280]{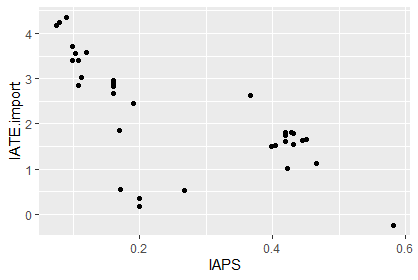}
		\subcaption{Import share (\%)}
	\end{minipage}
	\caption{Individual ATE plotted against individual average propensity score} \label{fig:scatplot}
\end{figure}

\begin{figure}[t]
	\begin{minipage}[b]{\textwidth}
		\centering
		\includegraphics[width=0.8\textwidth, bb = 0 0 710 253]{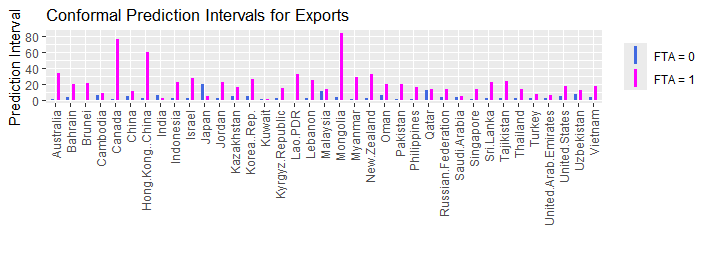}
		\subcaption{Export share (\%)} 
	\end{minipage}
	\begin{minipage}[b]{\textwidth}
		\centering
		\includegraphics[width=0.8\textwidth, bb = 0 0 710 253]{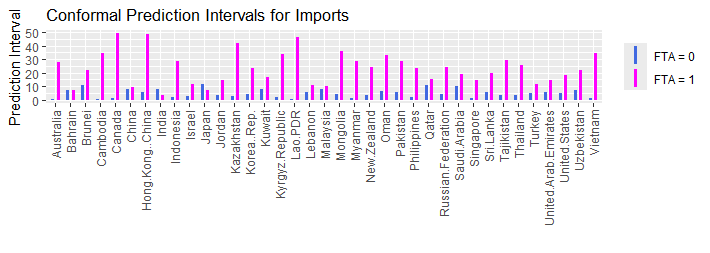}
		\subcaption{Import share (\%)}
	\end{minipage}
	\caption{Conformal prediction intervals ($\alpha = 0.1$)} \label{fig:ci}
\end{figure}

\section*{Declaration of the use of generative AI and AI-assisted technologies}

During the preparation of this work the authors used ChatGPT in order to improve readability and language.
After using this tool the authors reviewed and edited the content as necessary and take full responsibility for the content of the publication.

\section*{Acknowledgements}

The authors thank the Editor (Paul Fearnhead), the Associate Editor, and two anonymous referees for their insightful comments and suggestions.
In particular, the authors appreciate the comments provided by one of the referees regarding Assumption \ref{as:PO} and Lemma \ref{lem:J12}.
This work was supported by JSPS KAKENHI Grant Numbers 20K01597 and 24K04817.

\clearpage
\bibliographystyle{tandfx}
\bibliography{paper-ref}	
\clearpage

\appendix

\section{Proofs} \label{sec:proof}

\subsection{Proof of Proposition \ref{prop:DATE}}

By the law of iterated expectations, we have
\begin{align*}
	\text{pr}[ A_{ij} = 1 \mid Y_{ij}(0), Y_{ij}(1), P_{ij} ]
	& = \bbE[ A_{ij} \mid Y_{ij}(0), Y_{ij}(1), P_{ij} ] \\
	& = \bbE \Big[ \bbE[ A_{ij} \mid Y_{ij}(0), Y_{ij}(1), U_i, U_j ] \; \Big| \; Y_{ij}(0), Y_{ij}(1), P_{ij} \Big] \\
	& = \bbE \Big[ \bbE[ A_{ij} \mid U_i, U_j ] \; \Big| \; Y_{ij}(0), Y_{ij}(1), P_{ij} \Big] \\
	& = P_{ij},
\end{align*}
implying the conditional independence between $(Y_{ij}(0), Y_{ij}(1))$ and $A_{ij}$ given $P_{ij}$.
Given this, it holds that
\begin{align*}
    m_{ij}(a)
    = \bbE[ Y_{ij} \mid A_{ij} = a, P_{ij}]
    = \bbE[ Y_{ij}(a) \mid A_{ij} = a, P_{ij}]
    = \bbE[ Y_{ij}(a) \mid P_{ij}]
    \quad \text{for $a \in \{0, 1\}$.}
\end{align*}
This completes the proof.
\qed 

\subsection{Proof of Theorem \ref{thm:convergence}}

For $a \in \{ 0, 1 \}$, let
\begin{align*}
    \mu_{ij}(a) 
    \equiv \bbE[ \bm{1}\{ A_{ij} = a \} Y_{ij} \mid P_{ij} ],
    \qquad 
    \hat \mu_{ij}(a) 
    \equiv \sum_{i' \in \mathcal{N}_i} w(i, i') \bm{1}\{ A_{i'j} = a \} Y_{i'j},
\end{align*}
where $\mathcal{N}_i = \{ i' \in [n]: i' \neq i, \tilde{d}(i, i') \le b \}$ and $w(i, i')$ is defined as in \eqref{eq:Ptilde}.
By construction, $m_{ij}(1) = \mu_{ij}(1) / P_{ij}$ and $m_{ij}(0) = \mu_{ij}(0) / (1 - P_{ij})$. 
Similarly, the following equalities hold by the definition of $\mathcal{N}_i$ and Assumptions \ref{as:kernel}(i): $\hat m_{ij}(1) = \hat \mu_{ij}(1) / \hat P_{ij}$ and $\hat m_{ij}(0) = \hat \mu_{ij}(0) / (1 - \hat P_{ij})$.

By the $c_r$ inequality,
\begin{align*}
	\frac{1}{n-1}\| \hat \tau_{i \cdot} - \tau_{i \cdot} \|_2^2
	& = \frac{1}{n-1} \left\| \big( \hat m_{i \cdot}(1) - m_{i \cdot}(1) \big) - \big( \hat m_{i \cdot}(0) - m_{i \cdot}(0) \big) \right\|_2^2 \\
	& \le \frac{2}{n-1} \sum_{a \in \{0,1\}} \| \hat m_{i \cdot}(a) - m_{i \cdot}(a) \|_2^2.
\end{align*}
Then, it suffices to show that uniformly in $i \in [n]$,
\begin{align*}
	\frac{1}{n-1} \| \hat m_{i \cdot}(a) - m_{i \cdot}(a) \|_2^2 
	= \frac{1}{n - 1} \sum_{j \neq i} \big( \hat m_{ij}(a) - m_{ij}(a) \big)^2
	\lesssim_{\bbP} \sqrt{ (\log n) / n }.
\end{align*}
We focus on the term for $a = 1$ and the proof for $a = 0$ is analogous.
Noting that
\begin{align*}
	\hat m_{ij}(1) - m_{ij}(1)
	= \frac{ \hat \mu_{ij}(1) }{ \hat P_{ij} } - \frac{ \mu_{ij}(1) }{ P_{ij }}
	= \frac{1}{ \hat P_{ij} } \big( \hat \mu_{ij}(1) - \mu_{ij}(1) \big) - \frac{ \mu_{ij}(1) }{ \hat P_{ij} P_{ij} } \big( \hat P_{ij} - P_{ij} \big),
\end{align*}
the $c_r$ inequality yields
\begin{align}
	& \frac{1}{n-1} \| \hat m_{i \cdot}(1) - m_{i \cdot}(1) \|_2^2 \nonumber \\
	& = \frac{1}{n - 1} \sum_{j \neq i} \left( \frac{1}{ \hat P_{ij} } \big( \hat \mu_{ij}(1) - \mu_{ij}(1) \big) - \frac{ \mu_{ij}(1) }{ \hat P_{ij} P_{ij} } \big( \hat P_{ij} - P_{ij} \big) \right)^2 \nonumber \\
	& \le \frac{2}{n - 1} \sum_{j \neq i} \left( \frac{1}{ \hat P_{ij} } \big( \hat \mu_{ij}(1) - \mu_{ij}(1) \big) \right)^2 + \frac{2}{n - 1} \sum_{j \neq i} \left(\frac{ \mu_{ij}(1) }{ \hat P_{ij} P_{ij} } \big( \hat P_{ij} - P_{ij} \big) \right)^2 \label{eq:m1decomp}.
\end{align}
For the first term of \eqref{eq:m1decomp}, uniformly in $i \in [n]$,
\begin{align*}
	\frac{1}{n - 1} \sum_{j \neq i} \left( \frac{1}{ \hat P_{ij} } \big( \hat \mu_{ij}(1) - \mu_{ij}(1) \big) \right)^2
	& \le \left( \max_{j \neq i} \frac{1}{ \hat P_{ij}^2 } \right) \left( \frac{1}{n - 1} \sum_{j \neq i} \big( \hat \mu_{ij}(1) - \mu_{ij}(1) \big)^2 \right) \\
	& \le \kappa_L^{-2} \left( \frac{1}{n - 1} \sum_{j \neq i} \big( \hat \mu_{ij}(1) - \mu_{ij}(1) \big)^2 \right) \\
	& \lesssim_{\bbP} \sqrt{ (\log n) / n },
\end{align*}
where the last line follows from Lemma \ref{lem:convergence}.
Similarly, for the second term of \eqref{eq:m1decomp}, uniformly in $i \in [n]$,
\begin{align*}
	\frac{1}{n - 1} \sum_{j \neq i} \left(\frac{ \mu_{ij}(1) }{ \hat P_{ij} P_{ij} } \big( \hat P_{ij} - P_{ij} \big) \right)^2
	& \le \left( \max_{j \neq i} \frac{ \mu_{ij}^2(1) }{ \hat P_{ij}^2 P_{ij}^2 } \right) \left( \frac{1}{n - 1} \sum_{j \neq i} \big( \hat P_{ij} - P_{ij} \big)^2 \right) \\
	& \le \frac{C_Y^2}{\kappa_L^2 \underline{C}_P^2} \left( \frac{1}{n - 1} \sum_{j \neq i} \big( \hat P_{ij} - P_{ij} \big)^2 \right) \\
	& \lesssim_{\bbP} \sqrt{ (\log n) / n },
\end{align*}
where the second line follows from Assumptions \ref{as:overlap} and \ref{as:PO}(ii), and the last line from Lemma \ref{lem:convergence}.
This completes the proof.
\qed

\begin{lemma} \label{lem:convergence}
	Under Assumptions \ref{as:treatmentmat} and \ref{as:kernel} -- \ref{as:Lipschitzpropensity score}, we have
	\begin{equation} \label{eq:ratepropensity-score1}
		\left[ d_{2,\infty} \big( \hat P, P \big) \right]^2 
		\lesssim_{\bbP} \sqrt{ (\log n) / n },
	\end{equation}
	where $\hat P = (\hat P_{ij})$ and $P = (P_{ij})$.
	If Assumptions \ref{as:independence1}, \ref{as:overlap}, \ref{as:PO}, and \ref{as:independence2} hold additionally, then
	\begin{equation} \label{eq:ratePO}
		\left[ d_{2,\infty} \big( \hat \mu(a), \mu(a) \big) \right]^2
		\lesssim_{\bbP} \sqrt{(\log n) / n}
        \qquad \text{for $a \in \{ 0, 1 \}$},
	\end{equation}
	where $\hat \mu(a) = (\hat \mu_{ij}(a))$ and $\mu(a) = (\mu_{ij}(a))$.
\end{lemma}

\begin{proof}
    To prove \eqref{eq:ratepropensity-score1}, the $c_r$ inequality yields
    \begin{align*}
        \left[ d_{2,\infty} \big( \hat P, P \big) \right]^2
        & = \max_{1 \le i \le n} \left[ \frac{1}{n - 1} \sum_{j \neq i} (\hat P_{ij} - P_{ij})^2 \right] \\
        & = \max_{1 \le i \le n} \left[ \frac{1}{n - 1} \sum_{j \neq i} \left( \frac{ \check P_{ij} + \check P_{ji} }{ 2 } - \frac{ P_{ij} + P_{ji} }{ 2 } \right)^2 \right] \\
        & \le \frac{1}{2} \max_{1 \le i \le n} \left[ \frac{1}{n - 1} \sum_{j \neq i} (\check P_{ij} - P_{ij})^2 + \frac{1}{n - 1} \sum_{j \neq i} (\check P_{ji} - P_{ji})^2 \right] \\
        & \le \frac{1}{2} \left[ d_{2,\infty} \big( \check P, P \big) \right]^2 + \frac{1}{2} \left[ d_{2,\infty} \big( \check P^\top, P^\top \big) \right]^2.
    \end{align*}
    Thus, \eqref{eq:ratepropensity-score1} follows once we show that
    \begin{align}
        \left[ d_{2,\infty} \big( \check P, P \big) \right]^2 
        & \lesssim_{\bbP} \sqrt{ (\log n) / n }, \label{eq:ratepropensity-score2} \\ 
        \left[ d_{2,\infty} \big( \check P^\top, P^\top \big) \right]^2 
        & \lesssim_{\bbP} \sqrt{ (\log n) / n }. \label{eq:ratepropensity-score3}
    \end{align}
    
    For the left-hand side of \eqref{eq:ratepropensity-score2}, the triangle inequality leads to 
    \begin{align*}
    	d_{2, \infty}(\check P, P) 
    	\le d_{2, \infty}(\check P, \tilde P) + d_{2, \infty}(\tilde P, P).
    \end{align*}
    For the first term of the right-hand side, since $	\check P_{ij} - \tilde P_{ij} = \bm{1}\{\tilde P_{ij} < \kappa_L\}(\kappa_L - \tilde P_{ij}) + \bm{1}\{\tilde P_{ij} > \kappa_U\}(\kappa_U - \tilde P_{ij})$, we have
    \begin{align*}
    	\frac{1}{n - 1} \sum_{j \neq i} (\check P_{ij} - \tilde P_{ij})^2
    	& = \frac{1}{n - 1} \sum_{j \neq i} \bm{1}\{\tilde P_{ij} < \kappa_L\}(\kappa_L - \tilde P_{ij})^2 + \frac{1}{n - 1} \sum_{j \neq i} \bm{1}\{\tilde P_{ij} > \kappa_U\}(\kappa_U - \tilde P_{ij})^2 \\
    	& \le \frac{1}{n - 1} \sum_{j \neq i} \bm{1}\{\tilde P_{ij} < \kappa_L\} + \frac{1}{n - 1} \sum_{j \neq i} \bm{1}\{\tilde P_{ij} > \kappa_U \}.
    \end{align*}
    To rewrite the last line, let $p = (p_{ij})$, where
    \begin{align*}
    	p_{ij} \equiv \sum_{i' \in \mathcal{N}_i}w(i, i') P_{i'j}. 
    \end{align*}
    Then, we have
    \begin{align*}
    	\bm{1}\{ \tilde P_{ij} < \kappa_L \}
    	& = \bm{1}\{ \tilde P_{ij} - p_{ij} < \kappa_L - p_{ij} \} \\
    	& \le \bm{1}\{ \tilde P_{ij} - p_{ij} < \kappa_L - \underline{C}_P \} \\
    	& = \bm{1}\{ p_{ij} - \tilde P_{ij} > \underline{C}_P - \kappa_L \} \\
    	& \le \frac{(\tilde P_{ij} - p_{ij})^2}{(\underline{C}_P - \kappa_L)^2},
    \end{align*}
    where the second line follows from the fact that Assumption \ref{as:overlap} implies $p_{ij} \in [ \underline{C}_P, \overline{C}_P ]$ and the last line follows from $\underline{C}_P - \kappa_L > 0$ by Assumption \ref{as:trim}.
    Similarly, we have
    \begin{align*}
    	\bm{1}\{ \tilde P_{ij} > \kappa_U \}
    	= \bm{1}\{ \tilde P_{ij} - p_{ij} > \kappa_U - p_{ij} \}
    	\le \bm{1}\{ \tilde P_{ij} - p_{ij} > \kappa_U - \overline{C}_P \}
    	\le \frac{ (\tilde P_{ij} - p_{ij})^2 }{ (\kappa_U - \overline{C}_P)^2 }.
    \end{align*}
    These results imply that 
    \begin{align*}
    	[d_{2, \infty}(\check P, \tilde P)]^2 
    	\lesssim [d_{2, \infty}(\tilde P, p)]^2.
    \end{align*}
    Thus, \eqref{eq:ratepropensity-score2} follows once we show that
    \begin{align}
    	\left[ d_{2,\infty} \big( \tilde P, P \big) \right]^2 
    	& \lesssim_{\bbP} \sqrt{ (\log n) / n }, \label{eq:ratepropensity-score4} \\ 
    	\left[ d_{2,\infty} \big( \tilde P, p \big) \right]^2 
    	& \lesssim_{\bbP} \sqrt{ (\log n) / n }. \label{eq:ratepropensity-score5} 
    \end{align}
    
	Based on the discussion in the last two paragraphs, the statement of this lemma follows if we show \eqref{eq:ratePO}, \eqref{eq:ratepropensity-score3}, \eqref{eq:ratepropensity-score4}, and \eqref{eq:ratepropensity-score5}.
	The rest of the proof establishes these results by building on the proof of Theorem 1 in \cite{zhang2017estimating}.
	For clarity of exposition, however, we focus on proving \eqref{eq:ratePO} for $a = 1$, \eqref{eq:ratepropensity-score4}, and \eqref{eq:ratepropensity-score5}, since the proofs for $a = 0$ and \eqref{eq:ratepropensity-score3} are analogous and can be safely omitted.
	
	To proceed, we introduce several useful preliminaries.
	First, letting $\theta_{ij}$ denote either $P_{ij}$ or $\mu_{ij}(1)$, the corresponding estimator $\tilde P_{ij}$ or $\hat \mu_{ij}(1)$ can be written as 
	\begin{align*}
		\hat \theta_{ij}
		\equiv \sum_{i' \in \mathcal{N}_i} w(i, i') A_{i'j} Z_{i'j}
		\qquad \text{with} \quad  
		w(i, i') 
		= \left. \mathcal{K}\left( \frac{\tilde d(i, i')}{b} \right) \middle/ \sum_{l \in \mathcal{N}_i} \mathcal{K}\left( \frac{\tilde d(i, l)}{b} \right) \right. ,
	\end{align*}
    where
    \begin{align*}
        Z_{i'j} \equiv 
        \begin{cases}
            1 & \text{if $\theta_{ij} = P_{ij}$,} \\
	 	Y_{i'j} & \text{if $\theta_{ij} = \mu_{ij}(1)$.}
	 \end{cases}
	\end{align*}
	Next, it is easy to see from the definition of $\mathcal{N}_i$ and the choice of $b$ in Assumption \ref{as:bw} that 
	\begin{align} \label{eq:numneighbors}
        |\mathcal{N}_i| \ge C_0 \sqrt{n \log n},
    \end{align}
	where $C_0$ is as in Assumption \ref{as:bw}.
	With this, the denominator of $w(i, i')$ is bounded below from $C_1 \sqrt{n \log n}$ for some constant $C_1 > 0$:
    \begin{align} \label{eq:denominator}
        \mathcal{D}_i 
        \equiv \sum_{l \in \mathcal{N}_i} \mathcal{K} \left( \frac{\tilde d(i, l)}{b} \right)
        \ge \underline{C}_{\mathcal{K}} |\mathcal{N}_i|
        \ge C_1 \sqrt{n \log n}, 
    \end{align}
	where the first inequality follows from Assumption \ref{as:kernel}(ii).
	Moreover, since $0 \le w(i, i') \le 1$ for all $i' \neq i$, Assumption \ref{as:kernel}(ii) and \eqref{eq:denominator} lead to
    \begin{align} \label{eq:wbound}
        w(i, i')
        = \mathcal{D}_i^{-1} \mathcal{K}\left( \frac{ \tilde d(i, i') }{ b } \right) 
        \lesssim \frac{1}{\sqrt{n \log n}}
    \end{align}
    uniformly in $i \in [n]$ and $i' \neq i$.
	
	We now proceed to prove \eqref{eq:ratePO} for $a = 1$ and \eqref{eq:ratepropensity-score4}, which requires showing that
	\begin{align*}
		\frac{1}{n - 1} \sum_{j \neq i} (\hat \theta_{ij} - \theta_{ij})^2
		\lesssim_{\bbP} \sqrt{(\log n) / n}.
	\end{align*}
	uniformly in $i \in [n]$.
	In what follows, all analyses hold uniformly in $i \in [n]$, and we hereafter omit this phrase for the sake of brevity.
	Using $\sum_{i' \in \mathcal{N}_i} w(i, i') = 1$ and the $c_r$ inequality, observe that
	\begin{align*}
        \frac{1}{n - 1} \sum_{j \neq i} ( \hat \theta_{ij} - \theta_{ij} )^2 
        & = \frac{1}{n - 1} \sum_{j \neq i} \left[ \sum_{i' \in \mathcal{N}_i} w(i, i') ( A_{i'j} Z_{i'j} - \theta_{ij} ) \right]^2 \\
        & = \frac{1}{n - 1} \sum_{j \neq i} \left[ \sum_{i' \in \mathcal{N}_i} w(i, i') ( A_{i'j} Z_{i'j} - \theta_{i'j} ) + \sum_{i' \in \mathcal{N}_i} w(i, i') ( \theta_{i'j} - \theta_{ij} ) \right]^2 \\
        & \le \frac{2}{n - 1} \sum_{j \neq i} \left[ \sum_{i' \in \mathcal{N}_i} w(i, i') ( A_{i'j} Z_{i'j} - \theta_{i'j} ) \right]^2 + \frac{2}{n - 1} \sum_{j \neq i} \left[ \sum_{i' \in \mathcal{N}_i} w(i, i') ( \theta_{i'j} - \theta_{ij} ) \right]^2 \\
        & \eqqcolon 2 J_1(i) + 2 J_2(i).
    \end{align*}
	Below, we show that $J_1(i) \lesssim_{\bbP} \sqrt{ (\log n) / n}$ and $J_2(i) \lesssim_{\bbP} \sqrt{ (\log n) / n}$, which lead to \eqref{eq:ratePO} for $a = 1$ and \eqref{eq:ratepropensity-score4}.
    In addition, note that, when $\theta_{ij} = P_{ij}$, the result for $J_1(i)$ alone implies \eqref{eq:ratepropensity-score5}. 
	 
    To bound $J_1(i)$, observe that
	\begin{align}
        J_1(i)
        & = \frac{1}{n-1} \sum_{j \neq i} \left[ \sum_{i' \in \mathcal{N}_i} w(i, i') ( A_{i'j} Z_{i'j} - \theta_{i'j} ) \right]^2 \nonumber \\
        & = \frac{1}{n - 1} \sum_{j \neq i} \sum_{i' \in \mathcal{N}_i} \sum_{i'' \in \mathcal{N}_i} w(i, i') w(i, i'') ( A_{i'j} Z_{i'j} - \theta_{i'j} ) ( A_{i''j} Z_{i''j} - \theta_{i''j} ) \nonumber \\
        & = \frac{1}{n - 1} \sum_{j \neq i} \sum_{i' \in \mathcal{N}_i} w^2(i, i') ( A_{i'j} Z_{i'j} - \theta_{i'j} )^2  \\
        & \quad + \frac{1}{n - 1} \sum_{j \neq i} \sum_{i' \in \mathcal{N}_i} \sum_{i'' \in \mathcal{N}_i, i'' \neq i'} w(i, i') w(i, i'') ( A_{i'j} Z_{i'j} - \theta_{i'j} ) ( A_{i''j} Z_{i''j} - \theta_{i''j} ) \\
		& \eqqcolon J_{11}(i) + J_{12}(i).
    \end{align}
	To bound $J_{11}(i)$, note that 
	\begin{align} \label{eq:boundAZ}
        | A_{ij} Z_{ij} - \theta_{ij}| = 
        \begin{cases}
            |A_{ij} - P_{ij}| \le 1 & \text{if $\theta_{ij} = P_{ij}$,} \\
	 	|A_{ij} Y_{ij} - \mu_{ij}(1)| \le 2 C_Y & \text{if $\theta_{ij} = \mu_{ij}(1)$,}
        \end{cases}
    \end{align}
	by Assumption \ref{as:PO}(ii).
	Then, letting $C_2 = 1 \vee 2C_Y$, we have
	\begin{align*}
        J_{11}(i)
        & \le \frac{C_2^2}{n - 1} \sum_{j \neq i} \sum_{i' \in \mathcal{N}_i} w^2(i, i') \\
        & \le \frac{C_2^2 \overline{C}_{\mathcal{K}}}{(n - 1) \mathcal{D}_i} \sum_{j \neq i} \sum_{i' \in \mathcal{N}_i} w(i, i') \\
        & \lesssim \frac{1}{\sqrt{ n \log n }},
	\end{align*}
	where the second inequality follows from Assumption \ref{as:kernel}(ii) and the last line follows from \eqref{eq:denominator} and $\sum_{i' \in \mathcal{N}_i} w(i, i') = 1$.
	To bound $J_{12}(i)$, Lemma \ref{lem:J12} shows that
	\begin{align} \label{eq:J12}
		J_{12}(i)
		\lesssim_{\bbP} \sqrt{(\log n) / n}.
	\end{align}
	Combining these results, we obtain
	\begin{align} \label{eq:J1}
		J_1(i)
		\lesssim_{\bbP} 1/\sqrt{n \log n} + \sqrt{(\log n) / n}.
	\end{align}
	
	To bound $J_2(i)$, observe that
	\begin{align*}
        J_2(i)
        & = \frac{1}{n - 1} \sum_{j \neq i} \left[ \sum_{i' \in \mathcal{N}_i} w(i, i') ( \theta_{i'j} - \theta_{ij} ) \right]^2 \\
        & \le \frac{1}{n - 1} \sum_{j \neq i} \sum_{i' \in \mathcal{N}_i} w(i, i') ( \theta_{i'j} - \theta_{ij} )^2 \\
        & = \sum_{i' \in \mathcal{N}_i} w(i, i')  \left[ \frac{1}{n - 1} \sum_{j \neq i} ( \theta_{i'j} - \theta_{ij} )^2 \right],
    \end{align*}
	by Jensen's inequality.
	When $\theta_{ij} = P_{ij}$, Lemma 2 of \cite{zhang2017estimating} shows that
	\begin{align*}
        \frac{1}{n - 1} \sum_{j \neq i} ( \theta_{i'j} - \theta_{ij} )^2
        & = \frac{1}{n - 1} \sum_{j \neq i} ( P_{i'j} - P_{ij} )^2 \\
        & \lesssim_{\bbP} \sqrt{(\log n) / n}
    \end{align*}
	uniformly in $i \in [n]$ and $i' \in \mathcal{N}_i$.
	When $\theta_{ij} = \mu_{ij}(1)$, uniformly in $i \in [n]$ and $i' \in \mathcal{N}_i$,
	\begin{align*}
        \frac{1}{n - 1} \sum_{j \neq i} ( \theta_{i'j} - \theta_{ij} )^2
        & = \frac{1}{n - 1} \sum_{j \neq i} \Big( \mu_{i'j}(1) - \mu_{ij}(1) \Big)^2 \\
        & = \frac{1}{n - 1} \sum_{j \neq i} \Big( \bbE[ A_{i'j} Y_{i'j}(1) \mid P_{i'j} ] - \bbE[ A_{ij} Y_{ij}(1) \mid P_{ij} ] \Big)^2 \\
        & = \frac{1}{n - 1} \sum_{j \neq i} \Big( P_{i'j} \bbE[ Y_{i'j}(1) \mid P_{i'j} ] - P_{ij} \bbE[ Y_{ij}(1) \mid P_{ij} ] \Big)^2 \\
        & = \frac{1}{n - 1} \sum_{j \neq i} \bigg( \Big( P_{i'j} - P_{ij} \Big) \bbE[ Y_{ij}(1) \mid P_{ij} ] + P_{i'j} \Big( \bbE[ Y_{i'j}(1) \mid P_{i'j} ] - \bbE[ Y_{ij}(1) \mid P_{ij} ] \Big) \bigg)^2 \\
        & \le \frac{2 C_Y^2}{n - 1} \sum_{j \neq i} \Big( P_{i'j} - P_{ij} \Big)^2 + \frac{2 \overline{C}_P^2}{n - 1} \sum_{j \neq i} \Big( \bbE[ Y_{i'j}(1) \mid P_{i'j} ] - \bbE[ Y_{ij}(1) \mid P_{ij} ] \Big)^2 \\
        & \le \frac{2 C_Y^2}{n - 1} \sum_{j \neq i} \Big( P_{i'j} - P_{ij} \Big)^2 + \frac{2 \overline{C}_P^2 L_Y^2}{n - 1} \sum_{j \neq i} \Big( P_{i'j} - P_{ij} \Big)^2 \\
        & \lesssim_{\bbP} \sqrt{(\log n) / n},
    \end{align*}
	where the third equality follows from Assumption \ref{as:independence1}, the first inequality follows from the $c_r$ inequality and Assumptions \ref{as:overlap} and \ref{as:PO}(ii), the second inequality follows from Assumption \ref{as:PO}(iii), and the last line follows from Lemma 2 of \cite{zhang2017estimating}.
    Thus, we obtain
    \begin{align}\label{eq:J2}
        J_2(i)
        \lesssim_{\bbP} \sqrt{(\log n) / n}.
    \end{align}
	Combining \eqref{eq:J1} and \eqref{eq:J2} yields the desired result.
\end{proof}

\bigskip 

The following lemma provides the bound for $J_{12}(i)$ in \eqref{eq:J12}, as used in the proof of Lemma \ref{lem:convergence}.

\begin{lemma} \label{lem:J12}
	Under the same assumptions as in Lemma \ref{lem:convergence},
	\begin{align*}
		J_{12}(i)
		\lesssim_{\bbP} \sqrt{ (\log n) / n }
	\end{align*}
	uniformly in $i \in [n]$.
\end{lemma}

\begin{proof}
	Observe that
	\begin{align*}
		| J_{12}(i) |
		& = \left| \frac{1}{n-1} \sum_{j \neq i} \sum_{i' \in \mathcal{N}_i} \sum_{i'' \in \mathcal{N}_i, i'' \neq i'} w(i, i') w(i, i'') ( A_{i'j} Z_{i'j} - \theta_{i'j} ) ( A_{i''j} Z_{i''j} - \theta_{i''j} ) \right| \\
		& \le \sum_{i' \in \mathcal{N}_i} \sum_{i'' \in \mathcal{N}_i, i'' \neq i'} w(i, i') w(i, i'') \left| \frac{1}{n-1} \sum_{j \neq i} ( A_{i'j} Z_{i'j} - \theta_{i'j} ) ( A_{i''j} Z_{i''j} - \theta_{i''j} ) \right| \\
		& = \frac{n}{n-1} \sum_{i' \in \mathcal{N}_i} \sum_{i'' \in \mathcal{N}_i, i'' \neq i'} w(i, i') w(i, i'') \left| \frac{1}{n} \sum_{j \neq i, i', i''} ( A_{i'j} Z_{i'j} - \theta_{i'j} ) ( A_{i''j} Z_{i''j} - \theta_{i''j} ) \right|,
	\end{align*}
	where the second line uses the triangle inequality and the last line follows from $A_{ii} = \theta_{ii} = 0$ for all $i \in [n]$.
	If we show that 
	\begin{align} \label{eq:J12-1}
		\max_{ (i', i''): i' \neq i'', i' \neq i, i'' \neq i } \left| \frac{1}{n} \sum_{j \neq i, i', i''} ( A_{i'j} Z_{i'j} - \theta_{i'j} ) ( A_{i''j} Z_{i''j} - \theta_{i''j} ) \right|
		\lesssim_{\bbP} \sqrt{ (\log n) / n },
	\end{align}
	then we obtain
	\begin{align*}
		| J_{12}(i) |
		& \lesssim_{\bbP} \sqrt{ (\log n) / n } \times \frac{n}{n-1} \sum_{i' \in \mathcal{N}_i} \sum_{i'' \in \mathcal{N}_i, i'' \neq i'} w(i, i') w(i, i'') \\
		& = \sqrt{ (\log n) / n } \times \frac{n}{n-1} \sum_{i'' \in \mathcal{N}_i, i'' \neq i'} w(i, i'') \\
		& = \sqrt{ (\log n) / n } \times \frac{n}{n-1} [ 1 - w(i, i') ] \\
		& \lesssim_{\bbP} \sqrt{ (\log n) / n },
	\end{align*}
	where the second and third lines follow from $\sum_{i' \in \mathcal{N}_i} w(i, i') = 1$ and the last line follows from \eqref{eq:wbound}.
	In what follows, we prove \eqref{eq:J12-1} by applying Hoeffding's inequality separately to the cases $\theta_{i'j} = P_{i'j}$ and $\theta_{i'j} = \mu_{i'j}(1)$.
	
	\bigskip 
	
	To prove \eqref{eq:J12-1} when $\theta_{i'j} = P_{i'j}$, recall that the indices $(i', i'', j)$ in \eqref{eq:J12-1} satisfy $i' \neq i''$, $i' \neq j$, and $i'' \neq j$.
	Then, by Assumption \ref{as:treatmentmat}, the summand in \eqref{eq:J12-1} satisfies
	\begin{align*}
		& \bbE[ ( A_{i'j} - P_{i'j} ) ( A_{i''j} - P_{i''j} ) \mid U_{i'}, U_{i''}, U_j ] \\
		& = \bbE \left[ \Big( \bm{1}\{ f(U_{i'}, U_j) \ge U_{i'j} \} - f(U_{i'}, U_j) \Big) \Big( \bm{1}\{ f(U_{i''}, U_j) \ge U_{i''j} \} - f(U_{i''}, U_j) \Big) \; \middle| \; U_{i'}, U_{i''}, U_j \right] \\
		& = {\bbE}_{( U_{i'j}, U_{i''j} )} \left[ \Big( \bm{1}\{ f(U_{i'}, U_j) \ge U_{i'j} \} - f(U_{i'}, U_j) \Big) \Big( \bm{1}\{ f(U_{i''}, U_j) \ge U_{i''j} \} - f(U_{i''}, U_j) \Big) \right] \\
		& = {\bbE}_{U_{i'j}} \Big[ \bm{1}\{ f(U_{i'}, U_j) \ge U_{i'j} \} - f(U_{i'}, U_j) \Big]  {\bbE}_{U_{i''j}} \Big[ \bm{1}\{ f(U_{i''}, U_j) \ge U_{i''j} \} - f(U_{i''}, U_j) \Big] \\
		& = \Big( f(U_{i'}, U_j) - f(U_{i'}, U_j) \Big) \Big( f(U_{i''}, U_j) - f(U_{i''}, U_j) \Big) \\
		& = 0.
	\end{align*}
	Here, the subscripts on $\mathbb{E}$ indicate that the expectations are taken with respect to the corresponding random vectors.
	By the law of iterated expectations, we obtain
	\begin{align*}
		\bbE[ ( A_{i'j} - P_{i'j} ) ( A_{i''j} - P_{i''j} ) \mid U_{i'}, U_{i''} ] = 0.
	\end{align*}
	Furthermore, conditional on $(U_{i'}, U_{i''})$, the random variable $( A_{i'j} - P_{i'j} ) ( A_{i''j} - P_{i''j} )$ is measurable with respect to $(U_j, U_{i'j}, U_{i''j})$ and is therefore independent across $j$ ($j \neq i', i''$) by Assumption \ref{as:treatmentmat}.
	Noting that $|( A_{i'j} - P_{i'j} ) ( A_{i''j} - P_{i''j} )| \le 1$, Hoeffding's inequality implies that for any $\epsilon > 0$,
	\begin{align*}
		\text{pr}\left( \left| \frac{1}{n} \sum_{j \neq i, i', i''} ( A_{i'j} - P_{i'j} ) ( A_{i''j} - P_{i''j} ) \right| \ge \epsilon \; \middle| \; U_{i'}, U_{i''} \right)
		& \le 2 \exp\left( - \frac{2 \epsilon^2}{ \sum_{j \neq i, i', i''} ( 2 n^{-1} )^2 } \right) \\
		& = 2 \exp\left( - \frac{n^2 \epsilon^2}{ 2 (n - 3) } \right) \\
		& \le 2 \exp\left( - \frac{n \epsilon^2}{ 2 } \right),
	\end{align*}
	where the last line follows from $(n - 3) < n$ for all $n \ge 4$.
	Hence, by Boole's inequality and the law of iterated expectations,
	\begin{align*}
		& \text{pr}\left( \max_{ (i', i''): i'' \neq i', i' \neq i, i'' \neq i } \left| \frac{1}{n} \sum_{j \neq i, i', i''} ( A_{i'j} - P_{i'j} ) ( A_{i''j} - P_{i''j} ) \right| \ge \epsilon \right) \\
		& \le \sum_{ (i', i''): i'' \neq i', i' \neq i, i'' \neq i } \text{pr}\left( \left| \frac{1}{n} \sum_{j \neq i, i', i''} ( A_{i'j} - P_{i'j} ) ( A_{i''j} - P_{i''j} ) \right| \ge \epsilon \right) \\
		& \le 2 n^2 \exp\left( - \frac{n \epsilon^2}{ 2 } \right).
	\end{align*}
	Setting $\epsilon = \sqrt{2 C_3 (\log n) / n}$ for a constant $C_3 > 2$ yields
	\begin{align*}
		2 n^2 \exp\left( - \frac{n \epsilon^2}{ 2 } \right)
		= 2 n^2 \exp(- C_3 \log n)
		= 2 n^{2 - C_3} 
		= o(1),
	\end{align*}
	which establishes \eqref{eq:J12-1} in the case where $\theta_{i'j} = P_{i'j}$.
	
	\bigskip 
	
	To prove \eqref{eq:J12-1} when $\theta_{i'j} = \mu_{i'j}(1)$, we rewrite the conditional mean of the summand in \eqref{eq:J12-1} as follows:
	\begin{align} \label{eq:J12-2}
	\begin{split}
		& \bbE \left[ \Big( A_{i'j} Y_{i'j} - \mu_{i'j}(1) \Big) \Big( A_{i''j} Y_{i''j} - \mu_{i''j}(1) \Big) \; \middle| \; U_{i'}, U_{i''}, U_j \right] \\
		& = \bbE [ A_{i'j} Y_{i'j} A_{i''j} Y_{i''j} \mid U_{i'}, U_{i''}, U_j ] - \mu_{i''j}(1) \bbE [ A_{i'j} Y_{i'j} \mid U_{i'}, U_{i''}, U_j ] \\
		& \quad - \mu_{i'j}(1) \bbE [ A_{i''j} Y_{i''j} \mid U_{i'}, U_{i''}, U_j ] + \mu_{i'j}(1) \mu_{i''j}(1).
	\end{split}
	\end{align}
	For the first term on the right-hand side of \eqref{eq:J12-2}, we have
	\begin{align*}
		\bbE [ A_{i'j} Y_{i'j} A_{i''j} Y_{i''j} \mid U_{i'}, U_{i''}, U_j ]
		& = \bbE[ A_{i'j} Y_{i'j}(1) A_{i''j} Y_{i''j}(1) \mid U_{i'}, U_{i''}, U_j ] \\
		& = \bbE \Big[ \bbE[ A_{i'j} Y_{i'j}(1) \mid U_{i'}, U_{i''}, U_j, U_{i''j}, \xi_{i''j}(1) ] A_{i''j} Y_{i''j}(1) \; \Big| \; U_{i'}, U_{i''}, U_j \Big] \\
		& = \bbE \Big[ \bbE[ A_{i'j} Y_{i'j}(1) \mid U_{i'}, U_j ] A_{i''j} Y_{i''j}(1) \; \Big| \; U_{i'}, U_{i''}, U_j \Big] \\
		& = \mu_{i'j}(1) \bbE[ A_{i''j} Y_{i''j}(1) \mid U_{i'}, U_{i''}, U_j ] \\
		& = \mu_{i'j}(1) \bbE[ A_{i''j} Y_{i''j}(1) \mid U_{i''}, U_j ] \\
		& = \mu_{i'j}(1) \mu_{i''j}(1),
	\end{align*}
	where the second line follows from the law of iterated expectations, the third line uses the conditional independence between $(U_{i'j}, \xi_{i'j}(1))$ and $(U_{i''}, U_{i''j}, \xi_{i''j}(1))$ given $(U_{i'}, U_j)$ by Assumption \ref{as:independence2}(i), the fourth and last lines follow from the fact that Assumptions \ref{as:independence1} and \ref{as:PO}(i) imply 
	\begin{align} \label{eq:injective}
		\begin{split}
			\bbE[ A_{i'j} Y_{i'j}(1) \mid U_{i'}, U_j ]
			= P_{i'j} \bbE[ Y_{i'j}(1) \mid U_{i'}, U_j ]
			= P_{i'j} \bbE[ Y_{i'j}(1) \mid P_{i'j} ]
			= \mu_{i'j}(1),
		\end{split}
	\end{align}
	and the fifth line uses the conditional independence between $(U_{i''j}, \xi_{i''j}(1))$ and $U_{i'}$ given $(U_{i''}, U_j)$ by Assumption \ref{as:independence2}(i).
	Similarly, for the second and third terms on the right-hand side of \eqref{eq:J12-2}, we have
	\begin{align*}
		\mu_{i''j}(1) \bbE [ A_{i'j} Y_{i'j} \mid U_{i'}, U_{i''}, U_j ] 
		= \mu_{i'j}(1) \bbE [ A_{i''j} Y_{i''j} \mid U_{i'}, U_{i''}, U_j ]
		= \mu_{i'j}(1) \mu_{i''j}(1).
	\end{align*}
	Thus, the left-hand side of \eqref{eq:J12-2} is equal to 0, and the law of iterated expectations leads to
	\begin{align*}
		\bbE \left[ \Big( A_{i'j} Y_{i'j} - \mu_{i'j}(1) \Big) \Big( A_{i''j} Y_{i''j} - \mu_{i''j}(1) \Big) \; \middle| \; U_{i'}, U_{i''} \right]
		= 0.
	\end{align*}
	Furthermore, conditional on $(U_{i'}, U_{i''})$, the random variable $[ A_{i'j} Y_{i'j} - \mu_{i'j}(1) ] [ A_{i''j} Y_{i''j} - \mu_{i''j}(1) ]$ is measurable with respect to $(U_j, U_{i'j}, U_{i''j}, \xi_{i'j}(1),  \xi_{i''j}(1))$ and is therefore independent across $j$ ($j \neq i', i''$) by Assumption \ref{as:independence2}(ii).
	Noting that $|[ A_{i'j} Y_{i'j} - \mu_{i'j}(1) ] [ A_{i''j} Y_{i''j} - \mu_{i''j}(1) ]| \le 4 C_Y^2$ by Assumption \ref{as:PO}(ii), Hoeffding's inequality implies that for any $\epsilon > 0$,
	\begin{align*}
		\text{pr}\left( \left| \frac{1}{n} \sum_{j \neq i, i', i''} \Big( A_{i'j} Y_{i'j} - \mu_{i'j}(1) \Big) \Big( A_{i''j} Y_{i''j} - \mu_{i''j}(1) \Big) \right| \ge \epsilon \; \middle| \; U_{i'}, U_{i''} \right)
		& \le 2 \exp\left( - \frac{2 \epsilon^2}{ \sum_{j \neq i, i', i''} ( 8 C_Y^2 n^{-1} )^2 } \right) \\
		& = 2 \exp\left( - \frac{n^2 \epsilon^2}{ 32 C_Y^4 (n - 3) } \right) \\
		& \le  2 \exp\left( - \frac{n \epsilon^2}{ 32 C_Y^4 } \right) \\
		& \le  2 \exp\left( - \frac{n \epsilon^2}{ C_4 } \right),
	\end{align*}
	where the third line follows from $(n - 3) < n$ for all $n \ge 4$, and $C_4$ is a positive constant such that $C_4 > 32 C_Y^4$.
	Hence, by Boole's inequality and the law of iterated expectations,
	\begin{align*}
		& \text{pr}\left( \max_{ (i', i''): i'' \neq i', i' \neq i, i'' \neq i } \left| \frac{1}{n} \sum_{j \neq i, i', i''} \Big( A_{i'j} Y_{i'j} - \mu_{i'j}(1) \Big) \Big( A_{i''j} Y_{i''j} - \mu_{i''j}(1) \Big) \right| \ge \epsilon \right) \\
		& \le \sum_{ (i', i''): i'' \neq i', i' \neq i, i'' \neq i } \text{pr}\left( \left| \frac{1}{n} \sum_{j \neq i, i', i''} \Big( A_{i'j} Y_{i'j} - \mu_{i'j}(1) \Big) \Big( A_{i''j} Y_{i''j} - \mu_{i''j}(1) \Big) \right| \ge \epsilon \right) \\
		& \le 2 n^2 \exp\left( - \frac{n \epsilon^2}{ C_4 } \right).
	\end{align*}
	Setting $\epsilon = \sqrt{C_4 C_5 (\log n) / n}$ for a constant $C_5 > 2$ yields
	\begin{align*}
		2 n^2 \exp\left( - \frac{n \epsilon^2}{ C_4 } \right)
		\le 2 n^2 \exp(- C_5 \log n)
		= 2 n^{2 - C_5} 
		= o(1),
	\end{align*}
	which establishes \eqref{eq:J12-1} in the case where $\theta_{i'j} = \mu_{i'j}(1)$.
\end{proof}

\section{Permutation Inference} \label{sec:permutation}

As mentioned in Remark \ref{remark:permutation}, we consider testing the following sharp null hypothesis:
\begin{align*}
    \mathbb{H}_0: \; Y_{ij}(0) = Y_{ij}(1) \; \text{for all $(i, j)$}.
\end{align*}
Under $\mathbb{H}_0$, the observed outcome $Y_{ij}$ satisfies $Y_{ij} = Y_{ij}(0) = Y_{ij}(1)$, and hence the potential outcomes schedule $W \equiv \{ Y(0), Y(1) \}$ is imputable from the observed outcomes $Y$, where $Y(0) = (Y_{ij}(0))$, $Y(1) = (Y_{ij}(1))$, and $Y = (Y_{ij})$.
Consequently, $\mathbb{H}_0$ implies that $\tau_{ij} = 0$ for all $(i, j)$, so that the rejection of $\mathbb{H}_0$ indicates the significant presence of $\tau_{ij}$ for some $(i, j)$.

For testing $\mathbb{H}_0$, let $T(A, W)$ denote any real-valued test statistic, which can be computed under $\mathbb{H}_0$.
For example, we can consider the following test statistic:
\begin{align} \label{eq:teststat}
	T(A, W) = \frac{1}{n(n - 1)} \sum_{i \in [n]} \sum_{j \neq i} \hat \tau_{ij}^2.
\end{align}
Denote the set of all permutations of $[n]$ as $\bm{G}$, such that $|\bm{G}| = n!$.
For a given permutation $g \in \bm{G}$, let $gA = (A_{g(i)g(j)})$ be the permuted treatment matrix.
Then, the $p$-value for testing $\mathbb{H}_0$ is defined by 
\begin{align} \label{eq:pval}
    P(A, W)
    \equiv \frac{1}{n!} \sum_{g \in \bm{G}} \bm{1}\{ T(gA, W) \ge T(A, W) \}.
\end{align}
With a nominal significance level $\alpha \in [0, 1]$, we reject $\mathbb{H}_0$ at level $\alpha$ if $P(A, W) \le \alpha$.

The next theorem demonstrates the size control property of our test under the additional assumption of independence between $A$ and $Y(0)$. 
Note that this additional condition does not preclude endogeneity arising from the dependence between $A_{ij}$ and the treatment effect $Y_{ij}(1) - Y_{ij}(0)$. 
For example, the assumption accommodates selection on returns, where units with higher $Y_{ij}(1) - Y_{ij}(0)$ are more likely to receive the treatment. 
In the presence of more general forms of endogeneity, however, there is no guarantee that \eqref{eq:size} holds, and our permutation inference may not be valid.

\begin{theorem}[Size control]\label{thm:size}
    Suppose that Assumption \ref{as:treatmentmat} holds.
    Also assume that $A$ is independent of $Y(0)$.
    Under $\mathbb{H}_0$, we have
    \begin{align} \label{eq:size}
        \text{pr}[ P(A, W) \le \alpha \mid A \in \mathcal{S}_a, W ] \le \alpha
        \quad \text{for any $\alpha \in [0, 1]$ and $a \in \mathrm{supp}[A]$,}
    \end{align}
    where $\mathcal{S}_a \equiv \{ ga: g \in \bm{G} \}$ and the probability is taken with respect to the conditional distribution of $A$ given $A \in \mathcal{S}_a$ and $W$.
\end{theorem}

\begin{proof}
    First, note that $P(A, W)$ in \eqref{eq:pval} is computable because $T(gA, W)$ for any $g \in \bm{G}$ is imputable under $\mathbb{H}_0$.
    Then, denoting $\mathcal{S}_A = \{ gA: g \in \bm{G} \}$, we show that under $\mathbb{H}_0$
    \begin{align} \label{eq:rewrite-pval}
        P(A, W)
        = {\text{pr}}^* [ T(A^*, W) \ge T(A, W) \mid A^* \in \mathcal{S}_A, A, W ],
    \end{align}
    where $A^*$ denotes an independent copy of $A$ such that $A^*$, $A$, and $W$ are mutually independent, and ${\text{pr}}^*$ indicates the (conditional) probability taken with respect to the (conditional) distribution of $A^*$.
    Note that without the independence assumption between $A$ and $Y(0)$, it is not possible to consider such $A^*$ independent of $W$ under $\mathbb{H}_0$.
    The right-hand side of \eqref{eq:rewrite-pval} can be rewritten as
    \begin{align*}
        {\text{pr}}^* [ T(A^*, W) \ge T(A, W) \mid A^* \in \mathcal{S}_A, A, W ]
        & = {\text{pr}}^* [ T(A^*, W) \ge T(A, W) \mid A^* \in \mathcal{S}_A ] \\
        & = \frac{1}{n!} \sum_{g \in \bm{G}} \bm{1}\{ T(g A, W) \ge T(A, W) \} \\
        & = P(A, W),
    \end{align*}
    where the first equality follows from the mutual independence between $A^*$, $A$, and $W$ and the second equality follows from the exchangeability in Assumption \ref{as:treatmentmat}.

    For any fixed $a \in \mathrm{supp}[A]$, let $F_a(\cdot; W)$ denote the CDF of $-T(A, W)$ given $A \in \mathcal{S}_a$ and $W$.
    Given $A \in \mathcal{S}_a$ (so $\mathcal{S}_A = \mathcal{S}_a$ by construction), the right-hand side of \eqref{eq:rewrite-pval} can be written as
    \begin{align*}
        P(A, W)
        & = {\text{pr}}^* [ T(A^*, W) \ge T(A, W) \mid A^* \in \mathcal{S}_a, A, W ] \\
        & = {\text{pr}}^* [ T(A^*, W) \ge T(A, W) \mid A^* \in \mathcal{S}_a, W ] \\
        & = F_a(-T(A, W); W),
    \end{align*}
    where the second equality follows from the mutual independence between $A^*$, $A$, and $W$.
    The desired result \eqref{eq:size} follows from the fact that $-T(A, W)$ has the CDF $F_a(\cdot; W)$ conditional on $A \in \mathcal{S}_a$ and $W$.
\end{proof}

\section{Supplementary Materials for the Numerical Simulation} \label{sec:simulation2}

\subsection{Supplementary tables for Section \ref{sec:simulation1}}\label{app:table}

Tables \ref{tab:MC1} and \ref{tab:conformal} summarize the detailed numerical results of the simulation analysis conducted in Section \ref{sec:simulation1}.
    
\subsection{Additional simulation results}

Here, we present the supplementary simulation results that complement the discussion in Section \ref{sec:simulation1}.

To verify the impacts of Assumption \ref{as:kernel}(ii) that requires the kernel function $\mathcal{K}$ to be bounded away from zero, we examine how the simulation results change when we use the standard triangular and Epanechnikov kernels that vanish outside $[0, 1]$.
We consider the same simulation setup as in Section \ref{sec:simulation1} except that we use the standard triangular and Epanechnikov kernels.
The simulation results are given in Table \ref{tab:MC3}, where 4 and 5 in the column $\mathcal{K}$ indicate the standard triangular and Epanechnikov kernels, respectively.
Differently from the main simulation results in Table \ref{tab:MC1}, the average mean squared error for the dyadic ATE estimation in Table \ref{tab:MC3} is drastically worse, which highlights the importance of Assumption \ref{as:kernel}(ii).

Recall that we set the bandwidth parameter $b$ as the $h = \sqrt{(\log n) / n}$-th quantile of $\{ \tilde d(i, j): j \neq i \}$ for the main Monte Carlo experiment in Section \ref{sec:simulation1}.
To investigate how the performance of our method depends on this bandwidth selection, we report the simulation results when setting $h = 0.5 \sqrt{(\log n) / n}$ and $h = 2 \sqrt{(\log n) / n}$ in Tables \ref{tab:MC4} and \ref{tab:MC5}.
Comments similar to the main simulation results reported in Table \ref{tab:MC1} can apply, and interestingly, the largest bandwidth outperforms the other choices in this simulation setup.

To examine the validity and the testing power of permutation inference developed in Appendix \ref{sec:permutation}, Table \ref{tab:permutation} reports the rejection frequencies for each data generating process across nominal significance levels $\alpha \in \{ 0.1, 0.05, 0.01 \}$.
The rejection frequencies in the null setup are sufficiently close to their nominal levels, which corroborates the result in Theorem \ref{thm:size}.
In contrast, in the linear and quadratic DATE setups, the rejection frequencies are sufficiently high and tend to 100\% as $n$ increases, suggesting the consistency of our permutation test.

\begin{table}[p]
    \caption{Monte Carlo simulation results: Neighbourhood kernel smoothing estimation} \label{tab:MC1}
        \begin{center}
        \begin{footnotesize}
        \begin{tabular}[h]{rrrrrrrrrrrr}
            \toprule
            \multicolumn{2}{c}{DGP} & \multicolumn{2}{c}{ } & \multicolumn{2}{c}{DATE} & \multicolumn{2}{c}{IATE} & \multicolumn{2}{c}{GATE} & \multicolumn{2}{c}{PS} \\
            \cmidrule(l{3pt}r{3pt}){1-2} \cmidrule(l{3pt}r{3pt}){5-6} \cmidrule(l{3pt}r{3pt}){7-8} \cmidrule(l{3pt}r{3pt}){9-10} \cmidrule(l{3pt}r{3pt}){11-12}
            $Y$ & $P$ & $n$ & $\mathcal{K}$ & AvgBias & AvgMSE & AvgBias & AvgMSE & AvgBias & AvgMSE & AvgBias & AvgMSE \\
            \midrule
            1 & 1 & 40 & 1 & 0.19 & 34.54 & 0.19 & 0.89 & 0.19 & 0.40 & -1.73 & 1.76\\
             &  &  & 2 & 0.21 & 35.52 & 0.21 & 0.91 & 0.21 & 0.40 & -1.74 & 1.80\\
             &  &  & 3 & 0.21 & 35.99 & 0.21 & 0.92 & 0.21 & 0.40 & -1.74 & 1.81\\
            \cmidrule(l{3pt}r{3pt}){3-12}
             &  & 80 & 1 & 0.07 & 22.30 & 0.07 & 0.29 & 0.07 & 0.10 & -1.23 & 1.16\\
             &  &  & 2 & 0.07 & 22.82 & 0.07 & 0.29 & 0.07 & 0.10 & -1.22 & 1.16\\
             &  &  & 3 & 0.07 & 23.11 & 0.07 & 0.30 & 0.07 & 0.10 & -1.22 & 1.16\\
            \cmidrule(l{3pt}r{3pt}){2-12}
             & 2 & 40 & 1 & -0.29 & 44.48 & -0.29 & 1.18 & -0.29 & 0.40 & -1.75 & 1.46\\
             &  &  & 2 & -0.30 & 45.63 & -0.30 & 1.21 & -0.30 & 0.40 & -1.76 & 1.47\\
             &  &  & 3 & -0.30 & 46.08 & -0.30 & 1.22 & -0.30 & 0.41 & -1.76 & 1.47\\
            \cmidrule(l{3pt}r{3pt}){3-12}
             &  & 80 & 1 & 0.12 & 33.21 & 0.12 & 0.42 & 0.12 & 0.11 & -0.94 & 0.68\\
             &  &  & 2 & 0.11 & 34.00 & 0.11 & 0.43 & 0.11 & 0.12 & -0.95 & 0.71\\
             &  &  & 3 & 0.11 & 34.35 & 0.11 & 0.44 & 0.11 & 0.12 & -0.95 & 0.72\\
            \midrule 
            2 & 1 & 40 & 1 & 3.65 & 63.68 & 3.65 & 2.30 & 3.65 & 0.86 & -1.73 & 1.76\\
             &  &  & 2 & 3.60 & 65.34 & 3.60 & 2.32 & 3.60 & 0.86 & -1.74 & 1.80\\
             &  &  & 3 & 3.58 & 66.15 & 3.58 & 2.33 & 3.58 & 0.86 & -1.74 & 1.81\\
            \cmidrule(l{3pt}r{3pt}){3-12}
             &  & 80 & 1 & 2.39 & 40.90 & 2.39 & 0.85 & 2.39 & 0.23 & -1.23 & 1.16\\
             &  &  & 2 & 2.32 & 41.77 & 2.32 & 0.85 & 2.32 & 0.23 & -1.22 & 1.16\\
             &  &  & 3 & 2.30 & 42.27 & 2.30 & 0.85 & 2.30 & 0.23 & -1.22 & 1.16\\
            \cmidrule(l{3pt}r{3pt}){2-12}
             & 2 & 40 & 1 & 4.72 & 82.47 & 4.72 & 3.13 & 4.72 & 0.94 & -1.75 & 1.46\\
             &  &  & 2 & 4.33 & 84.12 & 4.33 & 3.10 & 4.33 & 0.91 & -1.76 & 1.47\\
             &  &  & 3 & 4.26 & 84.85 & 4.26 & 3.10 & 4.26 & 0.91 & -1.76 & 1.47\\
            \cmidrule(l{3pt}r{3pt}){3-12}
             &  & 80 & 1 & 1.17 & 60.59 & 1.17 & 1.21 & 1.17 & 0.22 & -0.94 & 0.68\\
             &  &  & 2 & 1.04 & 61.87 & 1.04 & 1.19 & 1.04 & 0.22 & -0.95 & 0.71\\
             &  &  & 3 & 1.02 & 62.45 & 1.02 & 1.19 & 1.02 & 0.22 & -0.95 & 0.72\\
            \midrule 
            3 & 1 & 40 & 1 & 5.94 & 67.76 & 5.94 & 3.32 & 5.94 & 1.10 & -1.73 & 1.76\\
             &  &  & 2 & 5.84 & 69.32 & 5.84 & 3.30 & 5.84 & 1.09 & -1.74 & 1.80\\
             &  &  & 3 & 5.81 & 70.12 & 5.81 & 3.31 & 5.81 & 1.08 & -1.74 & 1.81\\
            \cmidrule(l{3pt}r{3pt}){3-12}
             &  & 80 & 1 & 4.02 & 43.38 & 4.02 & 1.40 & 4.02 & 0.34 & -1.23 & 1.16\\
             &  &  & 2 & 3.91 & 44.14 & 3.91 & 1.37 & 3.91 & 0.33 & -1.22 & 1.16\\
             &  &  & 3 & 3.88 & 44.62 & 3.88 & 1.37 & 3.88 & 0.33 & -1.22 & 1.16\\
            \cmidrule(l{3pt}r{3pt}){2-12}
             & 2 & 40 & 1 & 7.90 & 87.54 & 7.90 & 4.89 & 7.90 & 1.37 & -1.75 & 1.46\\
             &  &  & 2 & 7.28 & 88.74 & 7.28 & 4.71 & 7.28 & 1.29 & -1.76 & 1.47\\
             &  &  & 3 & 7.16 & 89.40 & 7.16 & 4.69 & 7.16 & 1.27 & -1.76 & 1.47\\
            \cmidrule(l{3pt}r{3pt}){3-12}
             &  & 80 & 1 & 1.84 & 62.09 & 1.84 & 1.94 & 1.84 & 0.24 & -0.94 & 0.68\\
             &  &  & 2 & 1.64 & 63.26 & 1.64 & 1.85 & 1.64 & 0.24 & -0.95 & 0.71\\
             &  &  & 3 & 1.60 & 63.83 & 1.60 & 1.84 & 1.60 & 0.24 & -0.95 & 0.72\\
            \bottomrule
        \end{tabular}
        \end{footnotesize}
        \end{center}
    \begin{footnotesize}
        Note: All entries are multiplied by $100$.

        \medskip 
        
        Abbreviations: DATE (dyadic ATE), IATE (individual ATE), GATE (global ATE), PS (propensity score), AvgBias (average bias), AvgMSE (average mean squared error).
    
        $\text{DGP ($Y$)} \in \{\text{1 (null)}, \text{2 (linear dyadic ATE)}, \text{3 (quadratic dyadic ATE)}\}$.
    
        $\text{DGP ($P$)} \in \{\text{1 (stochastic block model)}, \text{2 (degree heterogeneity model)}\}$.
    
        $\mathcal{K} \in \{\text{1 (uniform)}, \text{2 (triangular)}, \text{3 (Epanechnikov)}\}$, with a baseline constant $\underline{C}_{\mathcal{K}} = 1$.
    \end{footnotesize}
\end{table}

\begin{table}[p]
	\caption{Monte Carlo simulation results: Row-wise full conformal inference ($\alpha = 0.1$)} \label{tab:conformal}
	\begin{center}
    \begin{footnotesize}
	\begin{tabular}[h]{rrrrrrr}
		\toprule
		\multicolumn{2}{c}{DGP} & & \multicolumn{2}{c}{$Y_{i,n+1}(0) | \{A_{i,n+1} = 0\}$} & \multicolumn{2}{c}{$Y_{i,n+1}(1) | \{A_{i,n+1} = 1\}$} \\
		\cmidrule(l{3pt}r{3pt}){1-2} \cmidrule(l{3pt}r{3pt}){4-5} \cmidrule(l{3pt}r{3pt}){6-7}
        $Y$ & $P$ & $n$ & AvgCoverage (\%) & AvgLength & AvgCoverage (\%) & AvgLength\\
		\midrule
		1 & 1 & 40 & 91.5 & 3.7 & 94.4 & 5.6\\
         &  & 80 & 90.8 & 3.5 & 91.9 & 3.8\\
        \cmidrule(l{3pt}r{3pt}){2-7}
         & 2 & 40 & 91.6 & 3.7 & 92.7 & 5.0\\
         &  & 80 & 90.9 & 3.5 & 91.8 & 3.8\\
        \midrule 
        2 & 1 & 40 & 91.5 & 3.7 & 94.2 & 6.8\\
         &  & 80 & 90.8 & 3.5 & 92.0 & 5.5\\
        \cmidrule(l{3pt}r{3pt}){2-7}
         & 2 & 40 & 91.6 & 3.7 & 93.4 & 6.4\\
         &  & 80 & 90.9 & 3.5 & 91.5 & 5.4\\
        \midrule 
        3 & 1 & 40 & 91.5 & 3.7 & 94.2 & 6.9\\
         &  & 80 & 90.8 & 3.5 & 91.9 & 5.5\\
        \cmidrule(l{3pt}r{3pt}){2-7}
         & 2 & 40 & 91.6 & 3.7 & 93.4 & 6.5\\
         &  & 80 & 90.9 & 3.5 & 91.6 & 5.4\\
		\bottomrule
	\end{tabular}
    \end{footnotesize}
    \end{center}
    \begin{footnotesize}
        Abbreviations: AvgCoverage (average coverage), AvgLength (average length).
    
        $\text{DGP ($Y$)} \in \{\text{1 (null)}, \text{2 (linear dyadic ATE)}, \text{3 (quadratic dyadic ATE)}\}$.
    
        $\text{DGP ($P$)} \in \{\text{1 (stochastic block model)}, \text{2 (degree heterogeneity model)}\}$.
    \end{footnotesize}
\end{table}

\begin{table}[p]
	\caption{Monte Carlo simulation results: Impacts of Assumption \ref{as:kernel}(ii)} \label{tab:MC3}
	\begin{center}
	\begin{footnotesize}
	\begin{tabular}[h]{rrrrrrrrrrrr}
		\toprule
		\multicolumn{2}{c}{DGP} & \multicolumn{2}{c}{ } & \multicolumn{2}{c}{DATE} & \multicolumn{2}{c}{IATE} & \multicolumn{2}{c}{GATE} & \multicolumn{2}{c}{PS} \\
		\cmidrule(l{3pt}r{3pt}){1-2} \cmidrule(l{3pt}r{3pt}){5-6} \cmidrule(l{3pt}r{3pt}){7-8} \cmidrule(l{3pt}r{3pt}){9-10} \cmidrule(l{3pt}r{3pt}){11-12}
		$Y$ & $P$ & $n$ & $\mathcal{K}$ & AvgBias & AvgMSE & AvgBias & AvgMSE & AvgBias & AvgMSE & AvgBias & AvgMSE \\
		\midrule
		1 & 1 & 40 & 4 & 0.18 & 106.45 & 0.18 & 2.71 & 0.18 & 0.39 & -1.98 & 3.71\\
         &  &  & 5 & 0.18 & 105.72 & 0.18 & 2.69 & 0.18 & 0.39 & -1.98 & 3.67\\
        \cmidrule(l{3pt}r{3pt}){3-12}
         &  & 80 & 4 & 0.04 & 64.36 & 0.04 & 0.81 & 0.04 & 0.11 & -1.06 & 1.87\\
         &  &  & 5 & 0.04 & 63.38 & 0.04 & 0.80 & 0.04 & 0.11 & -1.06 & 1.84\\
        \cmidrule(l{3pt}r{3pt}){2-12}
         & 2 & 40 & 4 & -0.33 & 87.19 & -0.33 & 2.26 & -0.33 & 0.41 & -1.69 & 2.38\\
         &  &  & 5 & -0.33 & 85.97 & -0.33 & 2.23 & -0.33 & 0.41 & -1.69 & 2.32\\
        \cmidrule(l{3pt}r{3pt}){3-12}
         &  & 80 & 4 & 0.12 & 58.99 & 0.12 & 0.75 & 0.12 & 0.12 & -0.86 & 1.16\\
         &  &  & 5 & 0.12 & 57.77 & 0.12 & 0.73 & 0.12 & 0.12 & -0.86 & 1.12\\
        \midrule 
        2 & 1 & 40 & 4 & 0.17 & 188.06 & 0.17 & 5.57 & 0.17 & 0.71 & -1.98 & 3.71\\
         &  &  & 5 & 0.18 & 186.92 & 0.18 & 5.55 & 0.18 & 0.71 & -1.98 & 3.67\\
        \cmidrule(l{3pt}r{3pt}){3-12}
         &  & 80 & 4 & 1.21 & 115.53 & 1.21 & 1.75 & 1.21 & 0.20 & -1.06 & 1.87\\
         &  &  & 5 & 1.22 & 113.89 & 1.22 & 1.73 & 1.22 & 0.20 & -1.06 & 1.84\\
        \cmidrule(l{3pt}r{3pt}){2-12}
         & 2 & 40 & 4 & 0.33 & 153.98 & 0.33 & 4.51 & 0.33 & 0.73 & -1.69 & 2.38\\
         &  &  & 5 & 0.36 & 152.15 & 0.36 & 4.47 & 0.36 & 0.73 & -1.69 & 2.32\\
        \cmidrule(l{3pt}r{3pt}){3-12}
         &  & 80 & 4 & 0.12 & 105.39 & 0.12 & 1.61 & 0.12 & 0.21 & -0.86 & 1.16\\
         &  &  & 5 & 0.14 & 103.39 & 0.14 & 1.59 & 0.14 & 0.21 & -0.86 & 1.12\\
        \midrule 
        3 & 1 & 40 & 4 & 1.05 & 195.16 & 1.05 & 6.65 & 1.05 & 0.75 & -1.98 & 3.71\\
         &  &  & 5 & 1.06 & 194.02 & 1.06 & 6.62 & 1.06 & 0.75 & -1.98 & 3.67\\
        \cmidrule(l{3pt}r{3pt}){3-12}
         &  & 80 & 4 & 2.22 & 118.78 & 2.22 & 2.21 & 2.22 & 0.24 & -1.06 & 1.87\\
         &  &  & 5 & 2.23 & 117.14 & 2.23 & 2.20 & 2.23 & 0.24 & -1.06 & 1.84\\
        \cmidrule(l{3pt}r{3pt}){2-12}
         & 2 & 40 & 4 & 1.26 & 158.73 & 1.26 & 5.52 & 1.26 & 0.76 & -1.69 & 2.38\\
         &  &  & 5 & 1.32 & 156.90 & 1.32 & 5.49 & 1.32 & 0.76 & -1.69 & 2.32\\
        \cmidrule(l{3pt}r{3pt}){3-12}
         &  & 80 & 4 & 0.27 & 106.91 & 0.27 & 2.08 & 0.27 & 0.21 & -0.86 & 1.16\\
         &  &  & 5 & 0.30 & 104.89 & 0.30 & 2.07 & 0.30 & 0.21 & -0.86 & 1.12\\
		\bottomrule
	\end{tabular}
	\end{footnotesize}
	\end{center}
    \begin{footnotesize}
        Note: All entries are multiplied by $100$.
        The same simulation setups are considered as in Table \ref{tab:MC1}, except that 4 and 5 in the $\mathcal{K}$ column indicate the standard triangular and Epanechnikov kernels, respectively, which vanish outside $[0, 1]$.
    \end{footnotesize}
\end{table}

\begin{table}[p]
	\caption{Monte Carlo simulation results: Bandwidth selection ($h = 0.5 \sqrt{ (\log n) / n }$)} \label{tab:MC4}
	\begin{center}
    \begin{footnotesize}
	\begin{tabular}[h]{rrrrrrrrrrrr}
		\toprule
		\multicolumn{2}{c}{DGP} & \multicolumn{2}{c}{ } & \multicolumn{2}{c}{DATE} & \multicolumn{2}{c}{IATE} & \multicolumn{2}{c}{GATE} & \multicolumn{2}{c}{PS} \\
		\cmidrule(l{3pt}r{3pt}){1-2} \cmidrule(l{3pt}r{3pt}){5-6} \cmidrule(l{3pt}r{3pt}){7-8} \cmidrule(l{3pt}r{3pt}){9-10} \cmidrule(l{3pt}r{3pt}){11-12}
		$Y$ & $P$ & $n$ & $\mathcal{K}$ & AvgBias & AvgMSE & AvgBias & AvgMSE & AvgBias & AvgMSE & AvgBias & AvgMSE \\
		\midrule
		1 & 1 & 40 & 1 & 0.25 & 60.00 & 0.25 & 1.52 & 0.25 & 0.41 & -1.69 & 2.25\\
         &  &  & 2 & 0.25 & 60.84 & 0.25 & 1.55 & 0.25 & 0.41 & -1.69 & 2.30\\
         &  &  & 3 & 0.25 & 61.28 & 0.25 & 1.56 & 0.25 & 0.41 & -1.69 & 2.32\\
        \cmidrule(l{3pt}r{3pt}){3-12}
         &  & 80 & 1 & 0.06 & 39.58 & 0.06 & 0.50 & 0.06 & 0.10 & -1.14 & 1.38\\
         &  &  & 2 & 0.05 & 40.16 & 0.05 & 0.51 & 0.05 & 0.10 & -1.13 & 1.40\\
         &  &  & 3 & 0.05 & 40.54 & 0.05 & 0.51 & 0.05 & 0.10 & -1.13 & 1.41\\
        \cmidrule(l{3pt}r{3pt}){2-12}
         & 2 & 40 & 1 & -0.33 & 71.72 & -0.33 & 1.88 & -0.33 & 0.41 & -1.69 & 1.86\\
         &  &  & 2 & -0.32 & 72.75 & -0.32 & 1.90 & -0.32 & 0.41 & -1.68 & 1.93\\
         &  &  & 3 & -0.32 & 73.25 & -0.32 & 1.92 & -0.32 & 0.41 & -1.68 & 1.95\\
        \cmidrule(l{3pt}r{3pt}){3-12}
         &  & 80 & 1 & 0.12 & 54.39 & 0.12 & 0.70 & 0.12 & 0.12 & -0.88 & 1.00\\
         &  &  & 2 & 0.11 & 55.12 & 0.11 & 0.70 & 0.11 & 0.12 & -0.88 & 1.03\\
         &  &  & 3 & 0.11 & 55.55 & 0.11 & 0.71 & 0.11 & 0.12 & -0.88 & 1.05\\
        \midrule
        2 & 1 & 40 & 1 & 2.75 & 109.94 & 2.75 & 3.33 & 2.75 & 0.81 & -1.69 & 2.25\\
         &  &  & 2 & 2.73 & 111.31 & 2.73 & 3.36 & 2.73 & 0.80 & -1.69 & 2.30\\
         &  &  & 3 & 2.73 & 112.05 & 2.73 & 3.37 & 2.73 & 0.80 & -1.69 & 2.32\\
        \cmidrule(l{3pt}r{3pt}){3-12}
         &  & 80 & 1 & 1.85 & 72.00 & 1.85 & 1.19 & 1.85 & 0.22 & -1.14 & 1.38\\
         &  &  & 2 & 1.82 & 72.97 & 1.82 & 1.20 & 1.82 & 0.22 & -1.13 & 1.40\\
         &  &  & 3 & 1.81 & 73.62 & 1.81 & 1.20 & 1.81 & 0.22 & -1.13 & 1.41\\
        \cmidrule(l{3pt}r{3pt}){2-12}
         & 2 & 40 & 1 & 1.23 & 129.08 & 1.23 & 3.92 & 1.23 & 0.75 & -1.69 & 1.86\\
         &  &  & 2 & 1.14 & 130.56 & 1.14 & 3.94 & 1.14 & 0.75 & -1.68 & 1.93\\
         &  &  & 3 & 1.12 & 131.33 & 1.12 & 3.96 & 1.12 & 0.75 & -1.68 & 1.95\\
        \cmidrule(l{3pt}r{3pt}){3-12}
         &  & 80 & 1 & 0.13 & 97.91 & 0.13 & 1.53 & 0.13 & 0.21 & -0.88 & 1.00\\
         &  &  & 2 & 0.10 & 99.08 & 0.10 & 1.53 & 0.10 & 0.21 & -0.88 & 1.03\\
         &  &  & 3 & 0.09 & 99.78 & 0.09 & 1.54 & 0.09 & 0.21 & -0.88 & 1.05\\
        \midrule 
        3 & 1 & 40 & 1 & 4.59 & 114.82 & 4.59 & 4.21 & 4.59 & 0.96 & -1.69 & 2.25\\
         &  &  & 2 & 4.55 & 116.14 & 4.55 & 4.22 & 4.55 & 0.95 & -1.69 & 2.30\\
         &  &  & 3 & 4.54 & 116.88 & 4.54 & 4.24 & 4.54 & 0.95 & -1.69 & 2.32\\
        \cmidrule(l{3pt}r{3pt}){3-12}
         &  & 80 & 1 & 3.16 & 74.63 & 3.16 & 1.66 & 3.16 & 0.29 & -1.14 & 1.38\\
         &  &  & 2 & 3.11 & 75.56 & 3.11 & 1.66 & 3.11 & 0.28 & -1.13 & 1.40\\
         &  &  & 3 & 3.09 & 76.19 & 3.09 & 1.66 & 3.09 & 0.28 & -1.13 & 1.41\\
        \cmidrule(l{3pt}r{3pt}){2-12}
         & 2 & 40 & 1 & 2.62 & 133.35 & 2.62 & 5.01 & 2.62 & 0.82 & -1.69 & 1.86\\
         &  &  & 2 & 2.47 & 134.72 & 2.47 & 5.00 & 2.47 & 0.82 & -1.68 & 1.93\\
         &  &  & 3 & 2.43 & 135.49 & 2.43 & 5.01 & 2.43 & 0.82 & -1.68 & 1.95\\
        \cmidrule(l{3pt}r{3pt}){3-12}
         &  & 80 & 1 & 0.30 & 99.36 & 0.30 & 2.02 & 0.30 & 0.21 & -0.88 & 1.00\\
         &  &  & 2 & 0.24 & 100.50 & 0.24 & 2.01 & 0.24 & 0.21 & -0.88 & 1.03\\
         &  &  & 3 & 0.23 & 101.20 & 0.23 & 2.01 & 0.23 & 0.21 & -0.88 & 1.05\\
		\bottomrule
	\end{tabular}
    \end{footnotesize}
	\end{center}
    \begin{footnotesize}
        Note: All entries are multiplied by $100$.
        The same simulation setups are considered as in Table \ref{tab:MC1}, except that the bandwidth is set to the $h = 0.5 \sqrt{ (\log n) / n }$-th quantile of $\{ \tilde d(i, j) \}_{j: j \neq i}$.
    \end{footnotesize}
\end{table}

\begin{table}[p]
	\caption{Monte Carlo simulation results: Bandwidth selection ($h = 2 \sqrt{ (\log n) / n }$)} \label{tab:MC5}
	\begin{center}
	\begin{footnotesize}
	\begin{tabular}[h]{rrrrrrrrrrrr}
		\toprule
		\multicolumn{2}{c}{DGP} & \multicolumn{2}{c}{ } & \multicolumn{2}{c}{DATE} & \multicolumn{2}{c}{IATE} & \multicolumn{2}{c}{GATE} & \multicolumn{2}{c}{PS} \\
		\cmidrule(l{3pt}r{3pt}){1-2} \cmidrule(l{3pt}r{3pt}){5-6} \cmidrule(l{3pt}r{3pt}){7-8} \cmidrule(l{3pt}r{3pt}){9-10} \cmidrule(l{3pt}r{3pt}){11-12}
		$Y$ & $P$ & $n$ & $\mathcal{K}$ & AvgBias & AvgMSE & AvgBias & AvgMSE & AvgBias & AvgMSE & AvgBias & AvgMSE \\
		\midrule
		1 & 1 & 40 & 1 & 0.23 & 19.30 & 0.23 & 0.51 & 0.23 & 0.37 & -1.55 & 1.54\\
         &  &  & 2 & 0.23 & 20.12 & 0.23 & 0.53 & 0.23 & 0.37 & -1.62 & 1.53\\
         &  &  & 3 & 0.23 & 20.38 & 0.23 & 0.54 & 0.23 & 0.38 & -1.62 & 1.53\\
        \cmidrule(l{3pt}r{3pt}){3-12}
         &  & 80 & 1 & 0.05 & 12.29 & 0.05 & 0.16 & 0.05 & 0.10 & -1.24 & 1.14\\
         &  &  & 2 & 0.05 & 12.69 & 0.05 & 0.17 & 0.05 & 0.10 & -1.24 & 1.10\\
         &  &  & 3 & 0.05 & 12.84 & 0.05 & 0.17 & 0.05 & 0.10 & -1.24 & 1.09\\
        \cmidrule(l{3pt}r{3pt}){2-12}
         & 2 & 40 & 1 & -0.39 & 19.19 & -0.39 & 0.52 & -0.39 & 0.35 & -1.47 & 2.27\\
         &  &  & 2 & -0.38 & 20.10 & -0.38 & 0.54 & -0.38 & 0.35 & -1.57 & 1.83\\
         &  &  & 3 & -0.38 & 20.27 & -0.38 & 0.55 & -0.38 & 0.35 & -1.58 & 1.80\\
        \cmidrule(l{3pt}r{3pt}){3-12}
         &  & 80 & 1 & 0.10 & 15.32 & 0.10 & 0.19 & 0.10 & 0.09 & -1.00 & 0.73\\
         &  &  & 2 & 0.10 & 15.98 & 0.10 & 0.20 & 0.10 & 0.09 & -1.00 & 0.65\\
         &  &  & 3 & 0.10 & 16.09 & 0.10 & 0.20 & 0.10 & 0.09 & -1.01 & 0.64\\
        \midrule
        2 & 1 & 40 & 1 & 4.46 & 36.44 & 4.46 & 1.82 & 4.46 & 0.85 & -1.55 & 1.54\\
         &  &  & 2 & 4.29 & 37.70 & 4.29 & 1.78 & 4.29 & 0.85 & -1.62 & 1.53\\
         &  &  & 3 & 4.26 & 38.14 & 4.26 & 1.79 & 4.26 & 0.85 & -1.62 & 1.53\\
        \cmidrule(l{3pt}r{3pt}){3-12}
         &  & 80 & 1 & 3.08 & 23.23 & 3.08 & 0.76 & 3.08 & 0.26 & -1.24 & 1.14\\
         &  &  & 2 & 2.95 & 23.80 & 2.95 & 0.73 & 2.95 & 0.26 & -1.24 & 1.10\\
         &  &  & 3 & 2.93 & 24.06 & 2.93 & 0.73 & 2.93 & 0.26 & -1.24 & 1.09\\
        \cmidrule(l{3pt}r{3pt}){2-12}
         & 2 & 40 & 1 & 10.52 & 38.33 & 10.52 & 3.02 & 10.52 & 1.70 & -1.47 & 2.27\\
         &  &  & 2 & 9.74 & 39.14 & 9.74 & 2.82 & 9.74 & 1.55 & -1.57 & 1.83\\
         &  &  & 3 & 9.65 & 39.38 & 9.65 & 2.81 & 9.65 & 1.53 & -1.58 & 1.80\\
        \cmidrule(l{3pt}r{3pt}){3-12}
         &  & 80 & 1 & 5.76 & 29.28 & 5.76 & 1.53 & 5.76 & 0.50 & -1.00 & 0.73\\
         &  &  & 2 & 4.84 & 29.92 & 4.84 & 1.32 & 4.84 & 0.41 & -1.00 & 0.65\\
         &  &  & 3 & 4.72 & 30.08 & 4.72 & 1.30 & 4.72 & 0.40 & -1.01 & 0.64\\
        \midrule
        3 & 1 & 40 & 1 & 7.32 & 40.65 & 7.32 & 3.21 & 7.32 & 1.21 & -1.55 & 1.54\\
         &  &  & 2 & 7.02 & 41.66 & 7.02 & 3.05 & 7.02 & 1.18 & -1.62 & 1.53\\
         &  &  & 3 & 6.98 & 42.08 & 6.98 & 3.04 & 6.98 & 1.17 & -1.62 & 1.53\\
        \cmidrule(l{3pt}r{3pt}){3-12}
         &  & 80 & 1 & 5.22 & 26.09 & 5.22 & 1.51 & 5.22 & 0.45 & -1.24 & 1.14\\
         &  &  & 2 & 5.00 & 26.44 & 5.00 & 1.41 & 5.00 & 0.42 & -1.24 & 1.10\\
         &  &  & 3 & 4.96 & 26.67 & 4.96 & 1.40 & 4.96 & 0.42 & -1.24 & 1.09\\
        \cmidrule(l{3pt}r{3pt}){2-12}
         & 2 & 40 & 1 & 17.55 & 47.28 & 17.55 & 6.69 & 17.55 & 3.69 & -1.47 & 2.27\\
         &  &  & 2 & 16.24 & 46.74 & 16.24 & 6.04 & 16.24 & 3.25 & -1.57 & 1.83\\
         &  &  & 3 & 16.09 & 46.88 & 16.09 & 6.01 & 16.09 & 3.20 & -1.58 & 1.80\\
        \cmidrule(l{3pt}r{3pt}){3-12}
         &  & 80 & 1 & 9.33 & 32.96 & 9.33 & 3.37 & 9.33 & 1.05 & -1.00 & 0.73\\
         &  &  & 2 & 7.80 & 32.77 & 7.80 & 2.81 & 7.80 & 0.79 & -1.00 & 0.65\\
         &  &  & 3 & 7.60 & 32.85 & 7.60 & 2.77 & 7.60 & 0.76 & -1.01 & 0.64\\
		\bottomrule
	\end{tabular}
    \end{footnotesize}
	\end{center}
    \begin{footnotesize}
        Note: All entries are multiplied by $100$.
        The same simulation setups are considered as in Table \ref{tab:MC1}, except that the bandwidth is set to the $h = 2 \sqrt{ (\log n) / n }$-th quantile of $\{ \tilde d(i, j) \}_{j: j \neq i}$.
    \end{footnotesize}
\end{table}

\begin{table}

\caption{Monte Carlo simulation results: Permutation inference}\label{tab:permutation}
    \begin{center}
    \begin{footnotesize}
    \begin{tabular}[h]{rrrrrrr}
    \toprule
    \multicolumn{2}{c}{DGP} & \multicolumn{2}{c}{ } & \multicolumn{3}{c}{Rejection frequency (\%)} \\
    \cmidrule(l{3pt}r{3pt}){1-2} \cmidrule(l{3pt}r{3pt}){5-7}
    $Y$ & $P$ & $n$ & $\mathcal{K}$ & $\alpha = 0.1$ & $\alpha = 0.05$ & $\alpha = 0.01$\\
    \midrule
    1 & 1 & 40 & 1 & 11.3 & 6.0 & 1.8\\
     &  &  & 2 & 10.7 & 6.2 & 1.8\\
     &  &  & 3 & 10.6 & 6.0 & 1.9\\
    \cmidrule(l{3pt}r{3pt}){3-7}
     &  & 80 & 1 & 8.8 & 3.7 & 0.4\\
     &  &  & 2 & 9.0 & 4.1 & 0.4\\
     &  &  & 3 & 9.1 & 4.0 & 0.4\\
    \cmidrule(l{3pt}r{3pt}){2-7}
     & 2 & 40 & 1 & 10.8 & 4.5 & 1.0\\
     &  &  & 2 & 10.7 & 4.7 & 0.9\\
     &  &  & 3 & 10.6 & 4.7 & 1.0\\
    \cmidrule(l{3pt}r{3pt}){3-7}
     &  & 80 & 1 & 9.4 & 3.6 & 0.7\\
     &  &  & 2 & 9.8 & 3.6 & 0.7\\
     &  &  & 3 & 9.8 & 3.6 & 0.7\\
    \midrule
    2 & 1 & 40 & 1 & 99.6 & 98.8 & 93.8\\
     &  &  & 2 & 99.5 & 98.9 & 94.0\\
     &  &  & 3 & 99.5 & 98.8 & 93.9\\
    \cmidrule(l{3pt}r{3pt}){3-7}
     &  & 80 & 1 & 100.0 & 100.0 & 100.0\\
     &  &  & 2 & 100.0 & 100.0 & 100.0\\
     &  &  & 3 & 100.0 & 100.0 & 100.0\\
    \cmidrule(l{3pt}r{3pt}){2-7}
     & 2 & 40 & 1 & 99.1 & 97.0 & 87.1\\
     &  &  & 2 & 98.7 & 96.2 & 86.3\\
     &  &  & 3 & 98.7 & 96.2 & 85.9\\
    \cmidrule(l{3pt}r{3pt}){3-7}
     &  & 80 & 1 & 100.0 & 100.0 & 100.0\\
     &  &  & 2 & 100.0 & 100.0 & 100.0\\
     &  &  & 3 & 100.0 & 100.0 & 100.0\\
    \midrule 
    3 & 1 & 40 & 1 & 100.0 & 99.8 & 99.1\\
     &  &  & 2 & 100.0 & 99.8 & 99.1\\
     &  &  & 3 & 100.0 & 99.8 & 99.2\\
    \cmidrule(l{3pt}r{3pt}){3-7}
     &  & 80 & 1 & 100.0 & 100.0 & 100.0\\
     &  &  & 2 & 100.0 & 100.0 & 100.0\\
     &  &  & 3 & 100.0 & 100.0 & 100.0\\
    \cmidrule(l{3pt}r{3pt}){2-7}
     & 2 & 40 & 1 & 100.0 & 99.8 & 98.4\\
     &  &  & 2 & 100.0 & 99.7 & 97.7\\
     &  &  & 3 & 100.0 & 99.7 & 97.5\\
    \cmidrule(l{3pt}r{3pt}){3-7}
     &  & 80 & 1 & 100.0 & 100.0 & 100.0\\
     &  &  & 2 & 100.0 & 100.0 & 100.0\\
     &  &  & 3 & 100.0 & 100.0 & 100.0\\
    \bottomrule
    \end{tabular}
    \end{footnotesize}
    \end{center}
    \begin{footnotesize}
        Note: The same simulation setups are considered as in Table \ref{tab:MC1}.
    \end{footnotesize}
\end{table}

\clearpage

\section{Treatment Spillover}\label{app:relax_sutva}

This appendix briefly discusses a possible extension of our framework that relaxes the no-interference assumption.

Suppose that, for each unit $i$, the primary partner $s(i)$ is known in advance.
For example, in the context of international trade, $s(i)$ could be the largest trading partner.
We focus on the case where interference for unit $i$ is driven exclusively by this partner $s(i)$, and other units do not cause spillover effects for unit $i$.
Specifically, we assume that the potential outcome given $A_{ij} = a$ and $A_{i s(i)} = b$ is well-defined for any possible $(a, b) \in \{ 0, 1 \}^2$ and dyad $(i, j)$, which we denote by $Y_{ij}(a, b)$.
If $s(i) = j$, it must hold that $a = b$, and the potential outcome here reduces to $Y_{ij}(a)$ as in the main text.
Then, the observed outcome for dyad $(i, j)$ is given by $Y_{ij} = Y_{ij}( A_{ij}, A_{i s(i)} )$.
We assume that the treatment exhibits the same graphon representation as in Assumption \ref{as:treatmentmat} in the main text:
\begin{align}
	A_{ij} &= \bm{1} \left\{ f(U_i, U_j) \ge U_{ij} \right\}, \\
	A_{i s(i)} &= \bm{1}\left\{ f(U_i, U_{s(i)}) \ge U_{i s(i)} \right\}.
\end{align}
The corresponding propensity scores are given by $P_{ij} = f(U_i, U_j)$ and $P_{i s(i)} = f(U_i, U_{s(i)})$, respectively.

The following proposition extends Proposition \ref{prop:DATE} to this setting.

\begin{proposition} \label{prop:spillover}
	Suppose that Assumptions \ref{as:treatmentmat} and \ref{as:overlap} hold.
	If $Y_{ij}(a, b)$ is conditionally independent of $(A_{ij}, A_{i s(i)})$ given $(U_i, U_j, U_{s(i)})$ for all dyads $(i, j)$ and all possible $(a, b) \in \{ 0, 1 \}^2$, then
	\begin{align}\label{eq:ident_sutva}
		\mathbb{E} [ Y_{ij}(a,b) \mid P_{ij}, P_{i s(i)}]
		= \mathbb{E} [ Y_{ij} \mid A_{ij} = a,\ A_{i s(i)} = b, P_{ij}, P_{i s(i)} ].
	\end{align}
\end{proposition}

\begin{proof}
	For notational simplicity, we write $s = s(i)$.
	By Assumption \ref{as:treatmentmat}, it is easy to see that
	\begin{align}
		\text{pr} ( A_{ij} = a, A_{is} = b \mid P_{ij}, P_{is} )
		& = \text{pr} ( A_{ij} = a, A_{is} = b \mid U_i, U_j, U_s ) \\
		& = P_{ij}^{a} ( 1 - P_{ij} )^{1-a} P_{is}^{b} ( 1 - P_{is} )^{1-b}.
	\end{align}
	By the law of iterated expectations and the conditional independence assumption, we obtain
	\begin{align}
		& \text{pr}\left(A_{ij}=a, A_{is}=b \mid Y_{ij}(a,b), P_{ij},P_{is}\right)  \\
		& \quad = \mathbb{E}\left[ \text{pr}\left(A_{ij}=a, A_{is}=b \mid Y_{ij}(a,b), U_i, U_j, U_s \right)\mid Y_{ij}(a,b), P_{ij}, P_{is} \right]  \\
		& \quad = \mathbb{E}\left[ \text{pr}\left(A_{ij}=a, A_{is}=b \mid U_i, U_j, U_s \right) \mid Y_{ij}(a,b), P_{ij}, P_{is} \right] \\
		& \quad = P_{ij}^{a} ( 1 - P_{ij} )^{1-a} P_{is}^{b} ( 1 - P_{is} )^{1-b},
	\end{align}
	implying that $Y_{ij}(a,b)$ is conditionally independent of $( A_{ij}, A_{is} )$ given $(P_{ij}, P_{is})$.
	Thus,
	\begin{align}
		\mathbb{E}\left[Y_{ij} \mid A_{ij}=a,\ A_{is}=b,\ P_{ij},P_{is}\right]
		& = \mathbb{E}\left[Y_{ij}(a,b)\mid A_{ij}=a,\ A_{is}=b,\ P_{ij},P_{is}\right] \\
		& = \mathbb{E}\left[Y_{ij}(a,b)\mid P_{ij},P_{is}\right].
	\end{align}
\end{proof}

Since the propensity scores $P_{ij}$ and $P_{i s(i)}$ can be estimated by the neighbourhood kernel smoothing method developed in Section \ref{sec:method}, Proposition \ref{prop:spillover} suggests that the conditional potential outcome mean $\mathbb{E} [ Y_{ij}(a,b) \mid P_{ij}, P_{i s(i)}]$ can also be estimated using an analogous approach.
Then, it would be possible to estimate the dyadic average direct effect and the dyadic average spillover effect:
\begin{align}
    \tau_{ij}^{\mathrm{direct}}(b)
    & \equiv
    \mathbb{E}\left[Y_{ij}(1,b)-Y_{ij}(0,b)\mid P_{ij},P_{is}\right], \\
    \tau_{ij}^{\mathrm{spillover}}(a)
    & \equiv
    \mathbb{E}\left[Y_{ij}(a,1)-Y_{ij}(a,0)\mid P_{ij},P_{is}\right].
\end{align}
Nevertheless, this extension substantially complicates the formal asymptotic theory and the choice of appropriate bandwidths.
Consequently, we leave a detailed investigation of this setting for future research.

\end{document}